\documentclass[11pt]{article}
\usepackage{wrapfig}
\usepackage[T1]{fontenc}
\usepackage{indentfirst}
\usepackage{threeparttable}
\usepackage[utf8]{inputenc}
\usepackage{amsmath}
\usepackage{amsfonts}
\usepackage{amssymb}
\usepackage{amsthm}
\usepackage{mathabx}
\usepackage{algorithm,algpseudocode,algorithmicx}
\usepackage{multirow}
\newtheorem{theorem}{Theorem}

\usepackage{xcolor}
\usepackage{longtable}
\usepackage{graphicx}
\usepackage{subcaption}
\usepackage{booktabs}
\usepackage{fmtcount} 
\usepackage{bm}
\usepackage[english]{babel}	
\usepackage{xparse}% http://ctan.org/pkg/xparse
\usepackage{natbib}
\usepackage{hyperref}
\usepackage{float}
\usepackage{siunitx}
\usepackage{authblk}
\usepackage[toc,page]{appendix}
\usepackage{setspace}
\usepackage{footnote}

\title{\textbf{Sparse Positive-Definite Estimation for Covariance Matrices with Repeated Measurements}}
\author{Sunpeng Duan}
\author{Guo Yu}
\author{Juntao Duan}
\author{Yuedong Wang}
\affil{\textit{Department of Statistics and Applied Probability, University of California, Santa Barbara, California 93106-3110, U.S.A.}}
\date{This version: June 7, 2023}

\def\T{{ \mathrm{\scriptscriptstyle T} }}

\DeclareMathOperator*{\argmin}{argmin}
\DeclareMathOperator*{\dprime}{\prime \prime}

\def\beqr{\begin{eqnarray}}
\def\eeqr{\end{eqnarray}}
\def\beqrs{\begin{eqnarray*}}
\def\eeqrs{\end{eqnarray*}}
\def\bc{\begin{center}}
\def\ec{\end{center}}

\newcommand{\veps}{\varepsilon}
\newcommand{\eps}{\epsilon}

\def\bY{\bm{Y}}

\newtheorem{lemma}{Lemma}
\newtheorem{corollary}{Corollary}

\newtheorem{model}{Model}

\algrenewcommand\algorithmicrequire{\textbf{Input:}}
\algrenewcommand\algorithmicensure{\textbf{Output:}}

\newcommand\keywords[1]{%
    \begingroup
    \footnotesize
    \noindent\textbf{Some Key Words:}\ #1
    \endgroup
}

\begin{document}

\maketitle

\begin{abstract}
Repeated measurements are common in many fields, where random variables are observed repeatedly across different subjects. Such data have an underlying hierarchical structure, and it is of interest to learn covariance/correlation at different levels. Most existing methods for sparse covariance/correlation matrix estimation assume independent samples. 
Ignoring the underlying hierarchical structure and correlation within the subject leads to erroneous scientific conclusions. In this paper, we study the problem of sparse and positive-definite estimation of between-subject and within-subject covariance/correlation matrices for repeated measurements. Our estimators are solutions to convex optimization problems that can be solved efficiently. We establish estimation error rates for the proposed estimators and demonstrate their favorable performance through theoretical analysis and comprehensive simulation studies. We further apply our methods to construct between-subject and within-subject covariance graphs of clinical variables from hemodialysis patients.
\end{abstract}


\section{Introduction}
\label{sec:intro}

Understanding the covariance structure among random variables is one of the most fundamental tasks in statistics with applications in a wide range of fields, including economics, biology, and biomedical sciences \citep{Bickel:Levina:2008b, Fan.etal:2016}.
Various sparse estimation methods have been proposed in the literature, especially for high-dimensional settings. However, virtually all current methods require the critical assumption of independent samples, which could be violated in many applications. This paper considers a special correlated data structure where observations are repeated measurements. 

In many fields, such as medicine, psychology, and neuroscience, random variables of interest are often measured repeatedly across different subjects, which leads to dependence among observations within each subject. For example, vital signs such as pulse and blood pressure are usually measured in multiple physical exams for each subject, and these measurements from the same subject are correlated. 
Conclusions drawn from ignoring such dependence structures among observations may be practically misguided or even erroneous \citep{Bae.etal:2016}.
Therefore, it is important to estimate covariance structures in the presence of dependence due to repeated measurements \citep{Ostroff:1993}.

Repeated measurements have an underlying hierarchical structure, and it is of scientific interest to define and estimate covariance structures at each level. In psychology, the nomothetic approach is used to study variations between subjects, and the idiographic approach is used to study variations within a subject \citep{Hamaker:2012}. Covariance structures between subjects and within a subject may thus be different. For example, physical activity tends to increase the heart rate of a person (within a subject), while physically active people tend to have a lower average heart
rate (between subjects) \citep{Epskamp.etal:2018}. This paper aims to develop new methods to estimate the within-subject and between-subject covariance structures simultaneously.

This paper uses the term ``subject'' to represent a generic experimental unit for simplicity. We consider a multivariate one-way random effect model for within-subject and between-subject covariance structures among $p$ random variables:
\beqr
\label{eq:simpleY}
\bY_{ij} = \bm{b}_i + \bm{\veps}_{ij}, \ j = 1,\ldots, n_i;  \ i=1,\ldots,m,
\eeqr
where $\bm{Y}_{ij}=(Y_{ij1},\ldots,Y_{ijp})^\T \in \mathbb{R}^p$ is the $j$-th  (out of $n_i$) observation of the $i$-th subject, $\bm{b}_i = (b_{i1},\cdots,b_{ip})^\T \in \mathbb{R}^p$ are independent and identically distributed random vectors with mean $\bm{0}$ and covariance matrix $\Sigma_b \in \mathbb{R}^{p \times p}$,  and $\bm{\veps}_{ij} = (\veps_{ij1}, \ldots, \veps_{ijp})^\T \in \mathbb{R}^p$ are independent and identically distributed random vectors with mean $\bm{0}$ and covariance matrix $\Sigma_{\veps} \in \mathbb{R}^{p \times p}$.
Additionally, $\bm{b}_i$ and $\bm{\veps}_{ij}$ are mutually independent.
The between-subject covariance $\Sigma_b$ measures the covariance structure among variables at the group level $E(\bY_{ij}\mid i)$. On the other hand, the within-subject covariance $\Sigma_{\veps}$ characterizes the covariance structure among components in $\bY_{ij}-E(\bY_{ij}\mid i)$. 
Model \eqref{eq:simpleY} has found wide applications, e.g., in the classical test theory \citep{Algina:Swaminathan:2015}, where the observed score is modeled as the summation of the true score (as a latent variable) and a random error.

For the cross-sectional data, which is a special case of \eqref{eq:simpleY} with $n_i=1$ for $i=1,\ldots,m$, it is clear that one can only estimate the overall covariance $\Sigma_b+\Sigma_\veps$, which does not separate the within-subject and between-subject covariance structures. 
When $n_i\ge 2$ for at least some $i\in\{1,\ldots,m\}$, a common approach is to use the aggregate data across subjects, $\{\bar{\bm{Y}}_{1\cdot},\ldots,\bar{\bm{Y}}_{m\cdot}\}$, where $\bar{\bm{Y}}_{i\cdot}=\sum_{j=1}^{n_i}\bm{Y}_{ij}/n_i$.
The sample covariance estimate based on this aggregated data, 
\beqr
\overline{\Sigma}=\frac{1}{m-1}\sum_{i=1}^m\left(\bar{\bm{Y}}_{i\cdot}-\frac{1}{m}\sum_{i=1}^m\bar{\bm{Y}}_{i\cdot}\right)\left(\bar{\bm{Y}}_{i\cdot}-\frac{1}{m}\sum_{i=1}^m\bar{\bm{Y}}_{i\cdot}\right)^\T,
\label{eq:naive_est}
\eeqr 
is an unbiased estimate of 
\beqr
E(\overline{\Sigma})=\Sigma_b+\sum_{i=1}^m\frac{1}{mn_i}\Sigma_\veps.
\label{eq:naive_exp}
\eeqr
Consequently, $\overline{\Sigma}$ is a biased estimate of either $\Sigma_\varepsilon$ or $\Sigma_b$. \cite{Epskamp.etal:2018} studied the between-subject covariance structure with the aggregated data. Statistical inferences based on aggregated data may be misinterpreted \citep{Fisher.etal:2018}. In particular, analysis based on aggregated data may result in an issue termed ecological fallacy or Simpson's paradox \citep{Piantadosi:1988,Freedman:1999,Hamaker:2012,Epskamp.etal:2018}. 

Furthermore, in high-dimensional settings where $p$ could be much larger than $m$ or $N$, the sample covariance estimate is no longer positive definite, making it less amenable for interpretation or downstream statistical tasks. To our best knowledge, there is no research on covariance structure learning for high-dimensional repeated measures data. We fill in this methodological gap in this paper by proposing two new sparse positive definite estimators, one for the within-subject covariance $\Sigma_\veps$ and one for the between-subject covariance matrix $\Sigma_b$. We demonstrate the benefit of our proposed estimators by comparing both theoretically and numerically with other estimators that have been previously studied in different settings.

\section{Sparse Estimation of $\Sigma_\veps$ and $\Sigma_b$}
\label{sec:method}

Most recent approaches to estimating a large covariance matrix involve regularized estimation based on an unbiased estimate of the target covariance matrix. In the setting with independent and identically distributed samples, it is straightforward to use the sample covariance matrix as an unbiased estimate, and methods in the literature differ in various approaches to imposing regularization. 
Specifically, methods based on thresholding the sample covariance matrix have been well-studied \citep{Bickel:Levina:2008a, Bickel:Levina:2008b, Cai:Yuan:2012}, and further improvements have been developed to ensure positive definiteness in the resulting estimates \citep{Rothman.etal:2009, Rothman:2012, Xue.etal:2012, Cui.etal:2016}. \citet{Bien:Tibshirani:2011} proposed a penalized likelihood procedure for estimating a sparse covariance matrix, which could be computationally intensive due to the non-convexity of the likelihood in the covariance matrix. 

There are several unbiased estimates of the two covariance matrices in the model \eqref{eq:simpleY}. We first consider the following unbiased estimates:
\begin{align}
\widehat{\Sigma}_{\veps} = (\sum_{i=1}^m n_i - m)^{-1}\sum_{i=1}^m\sum_{j=1}^{n_i}(\bY_{ij}-\bar{\bY}_{i\cdot})(\bY_{ij}-\bar{\bY}_{i\cdot})^\T, \qquad 
\widehat{\Sigma}_{b}=\overline{\Sigma}-\sum_{i=1}^m(mn_i)^{-1} \widehat{\Sigma}_\veps.
\label{eq:sample}
\end{align}
The sample estimate $\widehat{\Sigma}_{\veps}$ is an unbiased estimate of $\Sigma_\veps$ \citep{Rao:Heckler:1998}.
From \eqref{eq:naive_est}, $\widehat{\Sigma}_{b}$ is an unbiased estimate of $\Sigma_b$, and is a multivariate extension of the unweighted sum-of-squares estimator in \cite{Rao:Sylvestre:1984}. We will consider another commonly used unbiased estimate of $\Sigma_b$ in Section \ref{sec:compare}.

Note that $\widehat{\Sigma}_\veps$ in \eqref{eq:sample} may be singular in high-dimensional settings where $p>m$, and $\widehat{\Sigma}_b$ may not be positive semi-definite for any dimensions.
In particular, the diagonal elements in $\widehat{\Sigma}_b$ could be negative. 
To derive sparse and positive-definite estimates of $\Sigma_\veps$ and $\Sigma_b$, we follow \citet{Xue.etal:2012} and consider   
the following optimization problem for estimating a generic covariance matrix $\Sigma$ with input matrix $B$,
\beqr
\min_{\Sigma\succeq\delta I_p}\frac{1}{2}\|\Sigma-B\|_F^2+\lambda|\Sigma|_1,\label{obj:admm0}
\eeqr
where $\|\cdot\|_F$ is the Frobenius norm and $|\cdot|_1$ is the $\ell_1$-norm of the off-diagonal elements of the input matrix. The constraint $\Sigma \succeq \delta I_p$ imposes positive semi-definiteness on $\Sigma - \delta I_p$, which results in a positive definite solution to \eqref{obj:admm0} with a small value of $\delta > 0$. 
This positive definiteness constraint is essential to provide a usable and accurate estimate (see Table~\ref{tab:error} in Appendix~\ref{sec:2compares} for numeric evidence).
A solution to \eqref{obj:admm0} is simultaneously sparse, positive definite, and close to the input matrix $B$, which is usually set as an unbiased sample estimate. 
Let $\widehat{\Sigma}_\veps^+$ 
be the sparse and positive definite estimate of $\Sigma_\veps$ as the solution to \eqref{obj:admm0} with  $B=\widehat{\Sigma}_\veps$ and $\lambda=\lambda_\veps$,  
and $\widehat{\Sigma}_b^+$ be the sparse and positive definite estimates of $\Sigma_b$ as the solution to \eqref{obj:admm0} with $B=\widehat{\Sigma}_b$ and 
$\lambda=\lambda_b$. We study $\widehat{\Sigma}_\veps^+$ and $\widehat{\Sigma}_b^+$ both theoretically and numerically.
In addition, to illustrate the suboptimality of using group aggregation in estimating either covariance matrix, we further study $\overline{\Sigma}^+$, which is defined as the solution to \eqref{obj:admm0} with $B = \overline{\Sigma}$ and $\lambda=\lambda_0$. The theoretical tuning parameter values $\lambda_\veps$, $\lambda_b$, and $\lambda_0$ are discussed in Section \ref{sec:theo}.

The convex optimization problem \eqref{obj:admm0} can be written equivalently as
\beqr
\min_{\Sigma,\Theta}\left\{\frac{1}{2}\|\Sigma-B\|_F^2+\lambda|\Theta|_1:~\Sigma=\Theta,~\Sigma\succeq \delta I_p\right\},
\label{obj:admm_eqv}
\eeqr
which we solve using the alternating direction method of multipliers \citep{Boyd.etal:2010}. Specifically, the algorithm iteratively minimizes the following augmented Lagrangian 
\beqrs
\label{eq:lagrangian}
L(\Sigma,\Theta;\Lambda)=\frac{1}{2}\|\Sigma-B\|_{F}^2+\lambda|\Theta|_1+\langle \Lambda,\Sigma-\Theta\rangle+\frac{\rho}{2}\|\Sigma-\Theta\|_{F}^2,
\eeqrs
over $\Sigma$, $\Theta$, and the dual variable $\Lambda$ using the following updates until convergence:
\beqr
  &\Sigma  \gets \argmin_{\Sigma\succeq\delta I_p} L(\Sigma,\Theta;\Lambda)=\frac{1}{1+\rho}(B+\rho \Theta-\Lambda,\delta)_+,\label{eq:x}\\
  &\Theta  \gets \argmin_{\Theta} L(\Sigma,\Theta;\Lambda)=\mathcal{S}_{\lambda/\rho}\Big(\Sigma+\frac{1}{\rho}\Lambda\Big), \label{eq:z}\\ \nonumber
  &\Lambda  \gets \Lambda+\rho(\Sigma-\Theta). \label{eq:y}
\eeqr
The update in \eqref{eq:x} computes the projection onto a positive semi-definite cone, where $(A,\delta)_+ =\sum_{j=1}^p\max(\lambda_j,\delta)v_j^\T v_j$ for a generic matrix $A \in \mathbb{R}^{p \times p}$ with the eigendecomposition $A = \sum_{j=1}^p\lambda_jv_j^\T v_j$. 
The update in \eqref{eq:z} evaluates element-wise soft-thresholding operators, where $\{\mathcal{S}_{b} (A)\}_{jk} = \mathrm{sign}(A_{jk})\max(|A_{jk}|-b,0)$ for any matrix $A$ and scalar $b \geq 0$. We follow \citet{Boyd.etal:2010} for practical considerations in this algorithm, including the initial values, the stopping criterion, and the updating strategy for the optimization parameter $\rho$, and refer to Appendix~\ref{sec:algo.details} for further implementation details. 
This algorithm has been widely used in the literature on covariance estimation \citep[e.g.,][]{Bien:Tibshirani:2011, Xue.etal:2012} with well-established convergence analysis \citep{Nishihara.etal:2015}. The computational complexity of each update is dominated by the eigendecomposition in \eqref{eq:x}, which requires $O(p^3)$ operations. An approximate alternating direction method of multipliers \citep{Rontsis.etal:2022} could be used to improve the computational complexity by avoiding repeated eigendecompositions.

\section{Theoretical Properties}
\label{sec:theo}

\subsection{Notations and assumptions}
\label{sec:theo_notation}
In this section, we derive the finite-sample estimation error rate of our proposed estimators $\widehat{\Sigma}_\veps^+$ (in Section \ref{sec:error.r}) and $\widehat{\Sigma}_{b}^+$ (in Section \ref{sec:error.g}), and establish their asymptotic consistency. 
In comparison, we further establish that $\overline{\Sigma}^+$ is inconsistent in estimating $\Sigma_b$ due to a non-vanishing bias even with an infinite number of subjects, thus illustrating the pitfall of the sample estimator \eqref{eq:naive_est} based on the aggregated data. 

We observe $\bY_{ij} \in \mathbb{R}^p$, which is the $j$-th repeated measurement of the $i$-th subject for $j = 1,\ldots, n_i$ and $i = 1,\ldots,m$, following the model \eqref{eq:simpleY}, where $\bm{\veps}_{ij}$ and $\bm{b}_i$ are $p$-dimensional sub-Gaussian random vectors with the true within and between covariance $\mathrm{cov}(\bm{\veps}_{ij}) = \Sigma_\veps^0$ and $\mathrm{cov}(\bm{b}_i) = \Sigma_b^0$ respectively, and $\bm{b}_i$ and $\bm{\veps}_{ij}$ are mutually independent.
Let $N = \sum_{i = 1}^m n_i$ be the total number of observations.
We consider the following class of sparse covariance matrices:
\beqrs
\mathcal{U} (M, s) = \bigg\{\Sigma \in \mathbb{S}^{p \times p}_{++}  :~ \max_{k} \Sigma_{k,k}\le M,~\max_{k}\sum_{\ell = 1}^p 1 (\Sigma_{k,\ell} \neq 0 )\le s \bigg\},
\eeqrs
where $\mathbb{S}^{p \times p}_{++}$ is the set of all $p$-by-$p$ symmetric positive definite matrices, and $\Sigma_{{k,\ell}}$ is the $(k,\ell)$-th entry of $\Sigma$.
A matrix in $\mathcal{U} (M, s)$ has diagonals bound $M$ and maximum row-wise (and by symmetry, column-wise) sparsity level $s$.

\subsection{Estimation error rate for the within-subject covariance estimator}
\label{sec:error.r}

\begin{theorem}[Estimation Error Rate of $\widehat{\Sigma}_\veps^+$]\label{thm:max.r.c}
Assume that the true within-subject covariance matrix $\Sigma^0_{\veps} \in \mathcal{U}(M_{\veps}, s_{\veps})$. Let $\lambda_\veps =C_1(N\log p)^{1/2}/(N-m)$ be the value of the tuning parameter $\lambda$ in \eqref{obj:admm0} for a sufficiently large constant $C_1 > 0$. If $\log p\le N$, the proposed within-subject estimator $\widehat{\Sigma}_\veps^{+}$ satisfies
$$
\left\|\widehat{\Sigma}_\veps^{+}-\Sigma^0_{\veps}\right\|_F\le5\lambda_\veps
(p s_\veps)^{1/2} 
$$
with probability at least $1-4p^{-C_2}$, where $C_2 > 0$ only depends on $C_1$ and $M_\veps$.
\end{theorem}

The term $(ps_\veps)^{1/2}$ in the error rate above represents the overall sparsity of the true covariance matrix $\Sigma_\veps^0$. This dependence on sparsity level has also been noted in \citet{Rothman.etal:2008} and \citet{Xue.etal:2012} over slightly different matrix classes.
Notably, the estimation error rate does not depend on $M_\veps$ or on the exact values of $n_i$ for $i = 1, ..., m$. 
Instead, the effective sample size in $\lambda$ is $N^{1/2} - N^{-1/2} m$, which only depends on the total observation number $N$ and the number of subjects $m$. 

\textit{Remark 1.} When the number of subject $m$ is relatively small compared with the total number of observations $N$ in the scale of $m = o(N^{1/2})$, Theorem \ref{thm:max.r.c} implies that
$$
\left\|\widehat{\Sigma}_\veps^+-\Sigma^0_{\veps}\right\|_F
=O_P\left\{(ps_\veps N^{-1}\log p)^{1/2}\right\}, 
$$
where $X_n=O_P(a_n)$ means that for a set of random variables $X_n$ and a corresponding set of constants $a_n$, $X_n/a_n$ is bounded by a positive constant with probability approaching 1.
This rate coincides with those in \citet{Bickel:Levina:2008b, Rothman.etal:2008, Rothman.etal:2009, Cai:Liu:2011, Xue.etal:2012}, which are derived based on the assumption of independent and identically distributed observations.

\textit{Remark 2.} On the other hand, with $m = O(N)$, e.g., when the number of repeated measurements of each subject is bounded by a constant, Theorem \ref{thm:max.r.c} implies that  
$$
\left\|\widehat{\Sigma}_\veps^+-\Sigma^0_{\veps}\right\|_F=O_P\left\{(ps_\veps m^{-1}\log p)^{1/2})\right\}.
$$
In this scenario, $m$ plays the role of the effective sample size, and 
estimation consistency is achieved when $m$ approaches infinity.

\subsection{Estimation error rate for the between-subject covariance estimator}
\label{sec:error.g}

\begin{theorem}[Estimation Error Rate of $\widehat{\Sigma}_b^+$]\label{thm:max.g2.c}
Assume that the true between-subject covariance matrix $\Sigma^0_{b} \in \mathcal{U}(M_{b}, s_{b})$ and the true within-subject covariance matrix $\Sigma^0_{\veps} \in \mathcal{U}(M_{\veps}, s_{\veps})$. Let
\beqr
\lambda_b = C_1 \left(\frac{\log p}{m}\right)^{1/2} + C_2 \frac{\left( N\log p\right)^{1/2}}{(N - m) n^\ast} + \frac{M_b}{m} + \frac{M_\veps}{m n^\ast }
\label{eq:lambdab}
\eeqr
be the value of the tuning parameter $\lambda$ in \eqref{obj:admm0} for sufficiently large $C_1, C_2 > 0$, where
$n^\ast = m / \sum_{i = 1}^m n_i^{-1}$.
If $\log p\le m$, then the proposed between-subject estimator $\widehat{\Sigma}_b^+$ satisfies 
\beqrs
\left\|\widehat{\Sigma}_{b}^{+}-\Sigma^0_b\right\|_F\le10\lambda_b (p s_b)^{1/2} %\label{eq:thm.g2matrix.bound.c}
\eeqrs
with probability at least $1- 8p^{-C_3}$, where $C_3> 0$ only depends on $C_1$, $C_2$ and $\max(M_\veps,M_b)$. 
\end{theorem}

Unlike the estimation error rate for $\widehat{\Sigma}_\veps$ in Theorem \ref{thm:max.r.c}, the rate for $\widehat{\Sigma}_b$ depends on the values of $n_i$'s via the term $n^\ast$. A simple upper bound is $n^\ast \geq \min_i n_i$, which implies that the second term in $\lambda_b$ converges to 0 at a rate that is at least not slower than $\lambda_\veps$ in Theorem \ref{thm:max.r.c}.
The rate in $\lambda_b$ is dominated by $(m^{-1}\log p)^{1/2}$.

Recall from \eqref{eq:naive_exp} that $\overline{\Sigma}$ has a bias of $\Sigma_\veps/n^\ast$ in estimating $\Sigma_b$. In practice, $\overline{\Sigma}$ has been misused to provide a sample estimate for subsequent regularized estimation \citep{Epskamp.etal:2018}. We establish the following estimation error rate for $\overline{\Sigma}^+$, which is defined as the solution to \eqref{obj:admm0} with input sample matrix $B = \overline{\Sigma}$, to illustrate that the bias in the sample estimate is carried over to the regularized estimation.

\begin{theorem}[Estimation Error Rate of $\overline{\Sigma}^+$]\label{thm:max.a.c}
Assume that the true between-subject covariance matrix $\Sigma^0_{b} \in \mathcal{U}(M_{b}, s_{b})$ and the true within-subject covariance matrix $\Sigma^0_{\veps} \in \mathcal{U}(M_{\veps}, s_{\veps})$. Let
\beqrs
\lambda_0 = C_1 \left(\frac{\log p}{m}\right)^{1/2} 
+ \frac{M_b}{m} +   \frac{M_\veps}{n^\ast}
\eeqrs
be the value of the tuning parameter $\lambda$ in \eqref{obj:admm0} for sufficiently large $C_1 > 0$, and the same $n^\ast$ defined in Theorem \ref{thm:max.g2.c}.
If $\log p\le m$, then the aggregated between-subject estimator $\overline{\Sigma}^+$ satisfies 
$$
\left\|\overline{\Sigma}^{+}-\Sigma^0_b\right\|_F\le10\lambda_0 (ps_b)^{1/2} 
$$
with probability at least $1- 4p^{-C_2}$, where $C_2> 0$ only depends on $C_1$ and $\max(M_\veps,M_b)$. 
\end{theorem}
The upper bound of the estimation error rate in $\overline{\Sigma}^+$ is strictly larger than that of $\widehat{\Sigma}_b^+$ due to the dominant term $M_\veps / n^\ast$ in $\lambda_0$, which corresponds to the bias in \eqref{eq:naive_exp}. For example, in the balanced setting where $n_i = n_1$ for all $i = 1, \ldots, m$, it holds that $n_\ast = n_1$ and this bias term $M_\veps / n_1$ does not vanish even if $m \rightarrow \infty$ as long as $n_1 = O(1)$. We also show that $\overline{\Sigma}^+$ is inconsistent in estimating $\Sigma_\veps$ in Theorem \ref{thm:max.ae.c} in Appendix \ref{sec:proofs}.

In some scenarios, the estimation of between-subject and within-subject correlation matrices, instead of covariance matrices, is of interest and can be obtained similarly in the proposed framework. We provide estimation error rates of the two sparse positive definite estimators in Appendix~\ref{sub:cor}. 

\section{Comparison between two unbiased estimators of $\Sigma_b$}
\label{sec:compare}
We consider a commonly used unbiased estimator of $\Sigma_b$ based on the multivariate analysis of variance \citep{Rao:Heckler:1998}:
\beqr
    \label{eq:Sigmatilde}
    \widetilde{\Sigma}_{b}= \frac{1}{n_0} \left\{ \sum_{i=1}^m \frac{n_i}{m - 1}(\bar{\bY}_{i\cdot}-\bar{\bY}_{\cdot\cdot})(\bar{\bY}_{i\cdot}-\bar{\bY}_{\cdot\cdot})^\T - \widehat{\Sigma}_\veps \right\}, \text{where  } 
    n_0= \frac{N- N^{-1}\sum_{i=1}^mn_i^2}{m-1},
\eeqr
$\bar{\bY}_{i \cdot}= n_i^{-1}\sum_{j=1}^{n_i}\bY_{ij}$,  $\bar{\bY}_{\cdot\cdot}=N^{-1}\sum_{i=1}^m\sum_{j=1}^{n_i}\bY_{ij}$, and $N = \sum_{i = 1}^m n_i$.

It is straightforward to show that $\mathrm{E} (\widetilde{\Sigma}_b) = \Sigma_b$. However, just like $\widehat{\Sigma}_b$ in \eqref{eq:sample}, the diagonal elements of $\widetilde{\Sigma}_b$ could be negative, which is undesirable for an estimate of $\Sigma_b$. Specifically, in the setting where $\bm{b}_i$ and $\bm{\veps}_{ij}$ follow Gaussian distributions and $n_i$'s are all equal, it can be shown that $\text{pr}\{(\widetilde{\Sigma}_{b})_{k,k}<0\}$ decreases with $(\Sigma^0_{b})_{k,k}/(\Sigma_{\veps}^0)_{k,k}$. An adjustment for negative diagonal values of $\widetilde{\Sigma}_b$ is proposed in \citet{Rao:Heckler:1998} based on the assumption that $\widehat{\Sigma}_\veps$ is positive definite, which is violated in the high-dimensional settings.

We demonstrate an additional limitation of using $\widetilde{\Sigma}_{b}$, in comparison with $\widehat{\Sigma}_b$, in obtaining a sparse positive definite estimate of $\Sigma_b$. Define $\widetilde{\Sigma}_b^+$ as a solution of \eqref{obj:admm0} with $B = \widetilde{\Sigma}_b$. The following theorem shows that the performance of $\widetilde{\Sigma}_b^+$ hinges on the data imbalance.
\begin{theorem}[Estimation Error Rate of $\widetilde{\Sigma}_b^+$]\label{thm:max.g.c}
Assume that the true between-subject covariance matrix $\Sigma^0_{b} \in \mathcal{U}(M_{b}, s_{b})$ and the true within-subject covariance matrix $\Sigma^0_{\veps} \in \mathcal{U}(M_{\veps}, s_{\veps})$. Let
$$
\widetilde{\lambda}_b = C_1 \frac{\max_i n_i}{n_0} \left(\frac{\log p}{m}\right)^{1/2} + C_2 \frac{\left( N\log p\right)^{1/2}}{n_0 (N - m)}
+ {\frac{(2N-n_0m)M_b}{2n_0 m}} + \frac{M_\veps}{n_0 m}
$$
be the value of the tuning parameter $\lambda$ in \eqref{obj:admm0} for sufficiently large $C_1, C_2 > 0$.
If $\log p\le m$, then $\widetilde{\Sigma}_b^+$ satisfies 
$$
\left\|\widetilde{\Sigma}_{b}^{+} - \Sigma^0_b\right\|_F \le 10 \widetilde{\lambda}_b (ps_b)^{1/2} %\label{eq:thm.g2matrix.bound.c}
$$
with probability at least $1- 8p^{-C_3}$, where $C_3> 0$ only depends on $C_1$, $C_2$ and $\max(M_\veps,M_b)$. 
\end{theorem}

We define a measure of data imbalance as $\max_i n_i / n_0 \geq 1$, where $n_0$ is defined in \eqref{eq:Sigmatilde}. 
In the balanced dataset where all $n_i$'s are equal, we have $ \max_i n_i / n_0 = 1$ and the two estimators coincide $\widehat{\Sigma}_b^+=\widetilde{\Sigma}_b^+$. This equivalence is also reflected by the same estimation error rate since $\lambda_b = \widetilde{\lambda}_b$. When $n_i$'s are not all equal, the imbalance $\max_i n_i / n_0 > 1$ increases with $\max_i n_i$ for fixed $m$ and $N$. Comparing the first term in $\lambda_b$ and $\widetilde{\lambda}_b$, the estimation error rate of $\widetilde{\Sigma}_b^+$ in the dimension $p$ is strictly worse than that of $\widehat{\Sigma}_b^+$, which does not depend on the imbalance of the dataset. We numerically verify this comparison in Section \ref{sec:numerical}, and demonstrate that the practical performance of $\widetilde{\Sigma}_b^+$ could be very sensitive to the imbalance of the data.

\section{Simulation studies}
\label{sec:numerical}
\subsection{General settings}
\label{sec:crrtiria}
In this section, we evaluate the numeric performance of our proposed estimators $\widehat{\Sigma}_\veps^+$ (for the within-subject covariance $\Sigma_\veps$) and $\widehat{\Sigma}_b^+$ (for the between-subject covariance $\Sigma_b$), and compare with $\overline{\Sigma}^+$ (in estimating either $\Sigma_b$ or $\Sigma_\veps$) and $\widetilde{\Sigma}_b^+$ (in estimating $\Sigma_b$).

In each of the subsequent subsections, we generate observations $\bY_{ij}$ from model \eqref{eq:simpleY}, where $\bm{b}_i\sim N(\bm{0},\Sigma_b^0)$ and $\bm{\veps}_{ij}\sim N(\bm{0},\Sigma_\veps^0)$. All estimators in comparison are defined as solutions to the optimization problem \eqref{obj:admm0} with corresponding input sample covariance matrices. We use a $5$-fold cross-validation procedure to select the optimal tuning parameter value $\lambda$ in \eqref{obj:admm0} for each problem. We refer to Appendix~\ref{sec:tunning} for the details of the cross-validation procedure.

\begin{figure}
\centering
\includegraphics[width=0.75\textwidth]{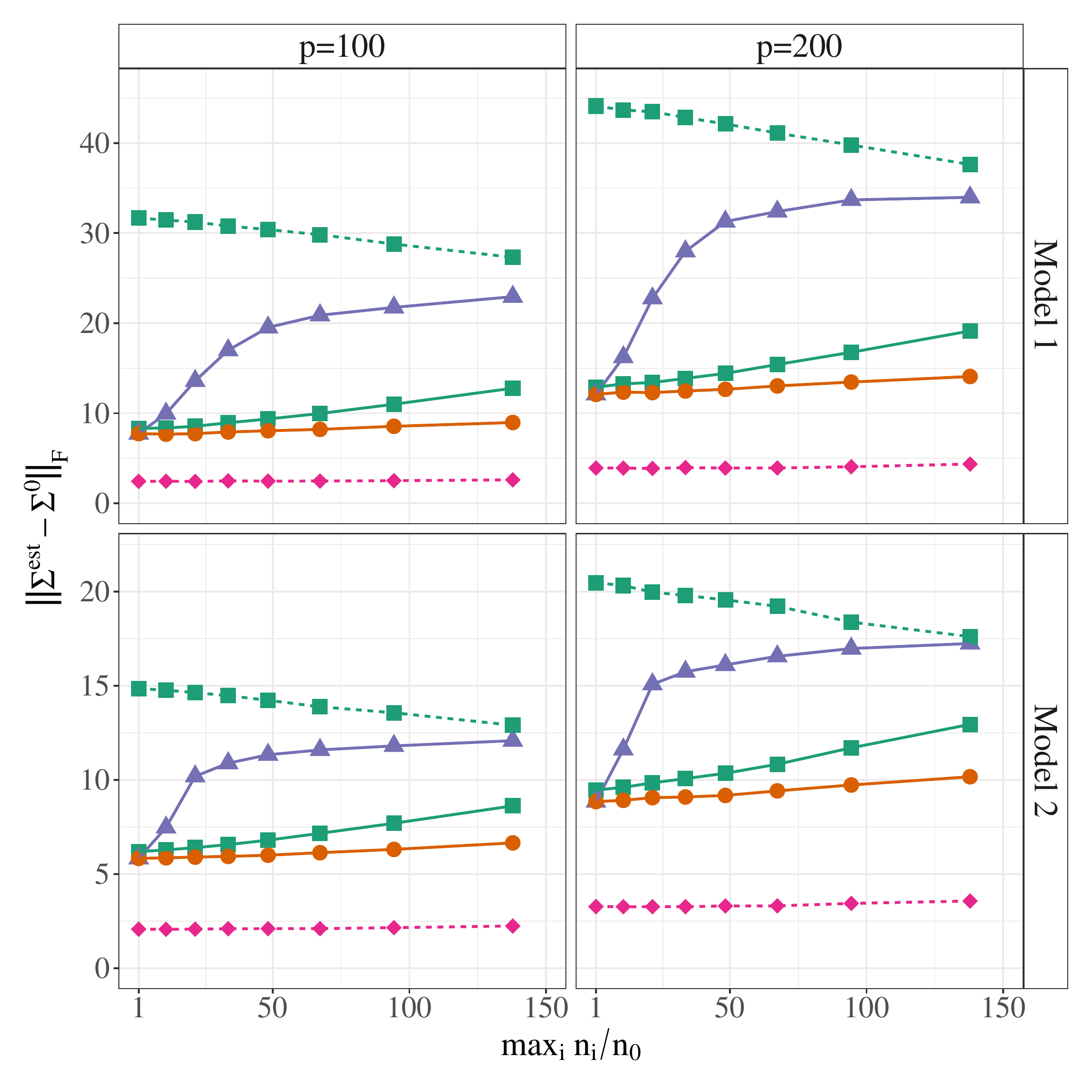}
\caption{Estimation error (in Frobenius norm, averaged over 100 replicates) for two between-subject (solid) and one within-subject (dash) covariance matrix estimator: $\widetilde{\Sigma}_b^+$ (violet triangle), $\widehat{\Sigma}_b^+$ (orange circle), and $\widehat{\Sigma}_\veps^+$ (pink diamond). The estimation error of the aggregated estimator ($\overline{\Sigma}^+$, green square) is evaluated in estimating the within-subject (dash) and the between-subject (solid) covariance matrices. The $x$-axis is $\max_in_i / n_0$, which characterizes the imbalance of the data.}
\label{fig:ferror}
\end{figure}

\subsection{General comparison}
In Sections \ref{sec:theo} and \ref{sec:compare}, we have shown that the error rates of the estimators depend on various factors: the number of subjects $m$, the total number of observations $N$, the ambient dimension $p$, and for $\widetilde{\Sigma}_b^+$ the data imbalance, i.e., $\max_in_i/n_0$. To numerically verify the established theoretical results, we consider the following models:
\begin{model}\label{model1} Banded matrices with bandwidth $10$: set $(\Sigma_b^0)_{j,k}= (1-|j-k|/10)_+$ and $(\Sigma_\veps^0)_{j, k} = (-1)^{|k_1-k_2|}(1-|k_1-k_2|/10)_+$;
\end{model}

\begin{model}\label{model2}
Covariance matrices corresponding to an \textsc{ar}$(1)$ series: set $(\Sigma_b^0)_{j, k} = 0.6^{|j - k|}$ and $(\Sigma_\veps^0)_{j, k} = (-0.6)^{|j - k|}$;
\end{model}
with $p=100$ and $p=200$. We note that the same covariance structures had been used in \citet{Bickel:Levina:2008a, Rothman:2012, Xue.etal:2012, Cui.etal:2016}.
In each setting, we let $N = 1000$ and $m = 100$. Furthermore, to study the effect of data imbalance on the estimation error, we set $n_i= a$ for $i=1,2,\dots,99$, where $a=\{3,4,\cdots,10\}$, and $n_{100} = N-99a$. By doing so, we generate settings where the measure of data imbalance, $\max_in_i /n_0$, varies.

Fig.~\ref{fig:ferror} summarizes the estimation error in the Frobenius norm averaged over 100 replications. We present the performance of four estimators: the proposed within-subject estimator $\widehat{\Sigma}_\veps^+$ for estimating $\Sigma_\veps^0$, and three between-subject estimators $\widehat{\Sigma}_{b}^+$ (our proposed method), $\widetilde{\Sigma}_{b}^+$ (the ANOVA type estimator), and $\overline{\Sigma}^+$ (the aggregated estimator) for estimating either $\Sigma_b^0$ or $\Sigma_\veps^0$. Among the three between-subject estimators, our proposed method $\widehat{\Sigma}_b^+$ achieves the lowest estimation error in all simulation settings. Furthermore, being consistent with the results in Theorem \ref{thm:max.g2.c}, Theorem \ref{thm:max.a.c} and Theorem \ref{thm:max.g.c}, the performance of $\widehat{\Sigma}_b^+$ and $\overline{\Sigma}^+$ are much less sensitive to the data imbalance $\max_in_i / n_0$ while the error of $\widetilde{\Sigma}_b^+$ dramatically increases as the data become less balanced. 
Surprisingly, in all but the perfectly balanced case ($\max_in_i  / n_0 = 1$), we observe that $\widetilde{\Sigma}^+_b$, which is built on the unbiased sample estimate \eqref{eq:Sigmatilde}, performs much worse than $\overline{\Sigma}^+$ which is built on the biased $\overline{\Sigma}$ in \eqref{eq:naive_est}. 
This suggests the dominating role of data imbalance in the estimation error of $\widetilde{\Sigma}_b^+$. 
Our proposed method $\widehat{\Sigma}_\veps^+$ also achieves much lower estimation errors than $\overline{\Sigma}^+$ in estimating within-subject covariance in all simulation settings.
The decreasing error of $\overline{\Sigma}^+$ in estimating $\Sigma_\veps^0$ is consistent with Theorem \ref{thm:max.ae.c} in Appendix \ref{sec:proofs}., which states that the error rate of $\|\overline{\Sigma}^+-\Sigma_\veps^0\|_F$ is inversely proportional to the imbalance score $\max_i n_i/n_0$.

\begin{figure}
\centering
\includegraphics[width=0.95\textwidth]{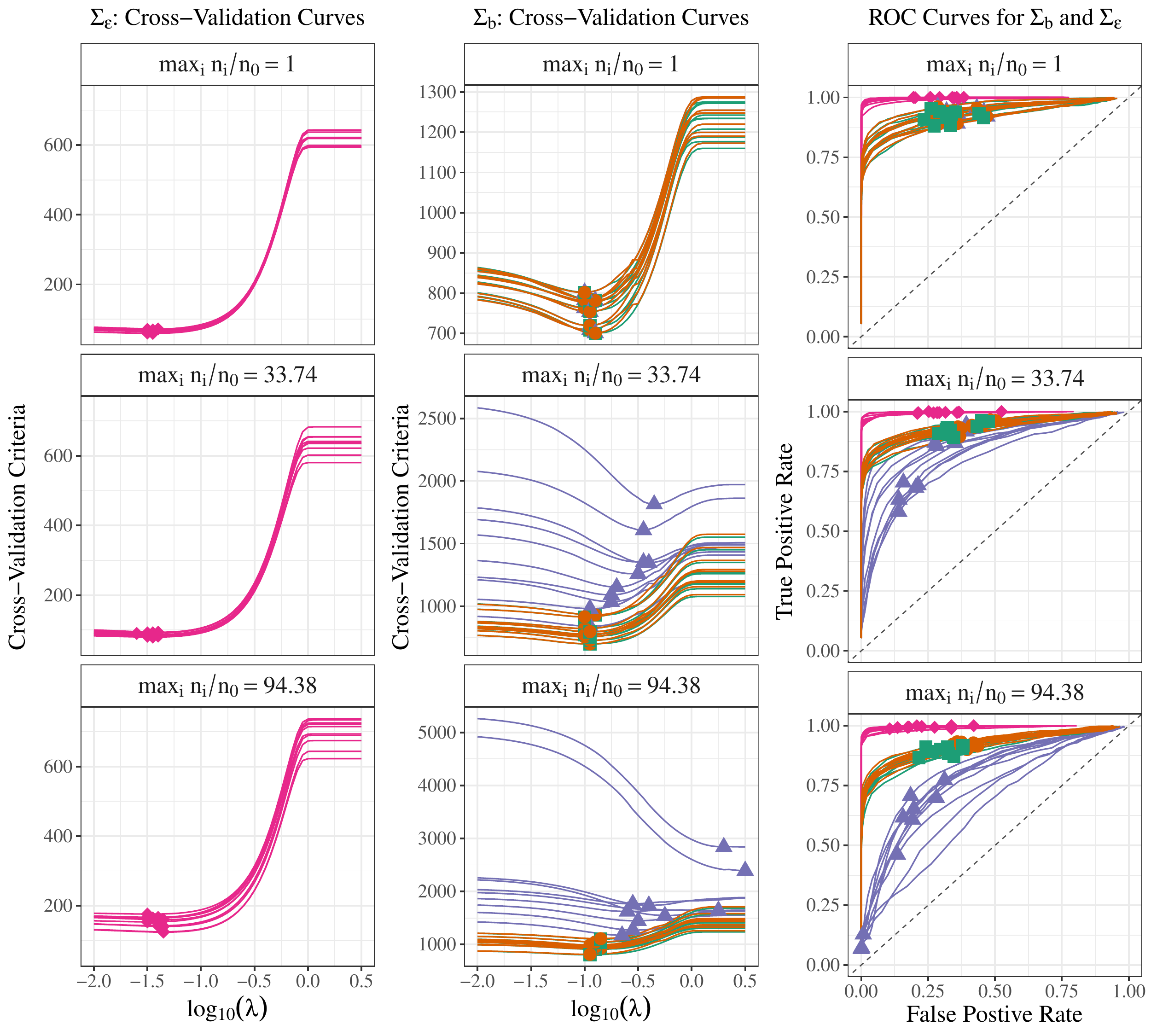}
% note that files may not be rotated
\caption{Cross-validation curves and receiver operating characteristic (ROC) curves between-subject and within-subject covariance sparsity recovery in Model \ref{model1} with $p=100$ and different values of $\max_in_i / n_0$. The top, middle and bottom rows correspond to different levels of data imbalance (with $a=10\text{, }7$, and $4$, respectively).  For simplicity of presentation, we randomly select 10 out of the 100 replicates. The left and middle panels exhibit 5-fold cross-validation curves of $\widehat{\Sigma}_\veps^+$ (pink) for within-subject covariance, $\widehat{\Sigma}_b^+$ (orange), $\widetilde{\Sigma}_b^+$ (violet), and $\overline{\Sigma}^+$ (green) for between-subject covariance. Diamonds ($\widehat{\Sigma}_\veps^+$), circles ($\widehat{\Sigma}_b^+$), triangles ($\widetilde{\Sigma}_b^+$), and squares ($\overline{\Sigma}^+$) in these two panels mark the minimum points on these curves. The right panels present the ROC curves. The diamonds ($\widehat{\Sigma}_\veps^+$), circles ($\widehat{\Sigma}_b^+$), triangles ($\widetilde{\Sigma}_b^+$), and squares ($\overline{\Sigma}^+$) represent the true positive rate and false positive rate with $\lambda$ values selected by the 5-fold cross-validation.}
\label{fig:model2_roc}
\end{figure}

To demonstrate the effectiveness of regularization, in Fig.~\ref{fig:model2_roc}, we present the cross-validation curves and the receiver operating characteristic (ROC) of the sparsity recovery of these estimators in Model \ref{model1} with $p=100$ and under three different levels of data imbalance. The optimal values of $\lambda$ for $\widehat{\Sigma}_{\veps}^+$, $\widehat{\Sigma}_{b}^+$, and $\overline{\Sigma}^+$ are relatively stable across different levels of data imbalance, while the optimal value of $\lambda$ for $\widetilde{\Sigma}_{b}^+$ sharply fluctuates and generally increases with $\max_in_i / n_0$. This indicates that large values of $\max_in_i / n_0$ tend to result in more shrinkage of the off-diagonal entries in $\widetilde{\Sigma}_{b}^+$ towards $0$. This observation is aligned with the larger error of $\widetilde{\Sigma}_b^+$ in Frobenius norm in Fig.~\ref{fig:ferror} for large values of $\max_in_i / n_0$. 

While the theoretical guarantees of support recovery would be an interesting and challenging problem for future research, we observe numerically that the data imbalance seems not to affect the support recovery performance of $\widehat{\Sigma}_{\veps}^+$, $\widehat{\Sigma}_{b}^+$, and $\overline{\Sigma}^+$, which is an established favorable properties of these estimators in terms of estimation error. In contrast, just as in estimation error, $\widetilde{\Sigma}_{b}^+$ suffers in sparsity recovery performance from the data imbalance.

\subsection{Understanding the effects of the bias in sample estimates}
As seen in Fig.~\ref{fig:ferror} and Fig.~\ref{fig:model2_roc}, the estimator $\overline{\Sigma}^+$ based on the biased sample estimate $\overline{\Sigma}$ surprisingly has relatively acceptable numerical performance.
This subsection investigates this observation by comparing our proposed between-subject estimator $\widehat{\Sigma}_b^+$ with $\overline{\Sigma}^+$. We consider two modifications of Model \ref{model1} as follows:

\begin{model}\label{model3}
For any given $a > 0$, we set
$(\Sigma_b^0)_{j, k} =(1-|j - k|/10)_+$ and 
$(\Sigma_\veps^0)_{j, k} = a(1-|j -k |/10)_+$.
\end{model}

\begin{model}\label{model4}
For any given $a > 0$, we set
$(\Sigma_b^0)_{j, k} =(1-|j - k|/10)_+$ and 
$(\Sigma_\veps^0)_{j, k} = a(-1)^{|j -k |}(1-|j -k |/10)_+$.
\end{model}
From \eqref{eq:naive_exp}, the matrix of $\Sigma_\veps$ can be considered as the noise for estimating $\Sigma_b$. We thus define the inverse signal-to-noise ratio as $|\Sigma_\veps^0|_\infty/|\Sigma_b^0|_\infty$. By varying $|\Sigma_\veps^0|_\infty/|\Sigma_b^0|_\infty=a\in \{1, 2, \ldots, 10\}$ in Model~\ref{model3} and Model~\ref{model4},  we construct settings where the relative signal strength from $\Sigma_\veps$ and $\Sigma_b$ is different. In comparison with Model \ref{model3}, we alternate the signs of sub-diagonal elements in $\Sigma_\veps^0$ in Model \ref{model4}. In both models, we generate balanced data with $n_i=5$ for $i = 1, \ldots, m=100$ and $p=50$. Estimation errors in Frobenius norm are summarized (over 100 replications) in Fig.~\ref{fig:comparsion2}.

\begin{figure}
\centering
\includegraphics[width=0.75\textwidth]{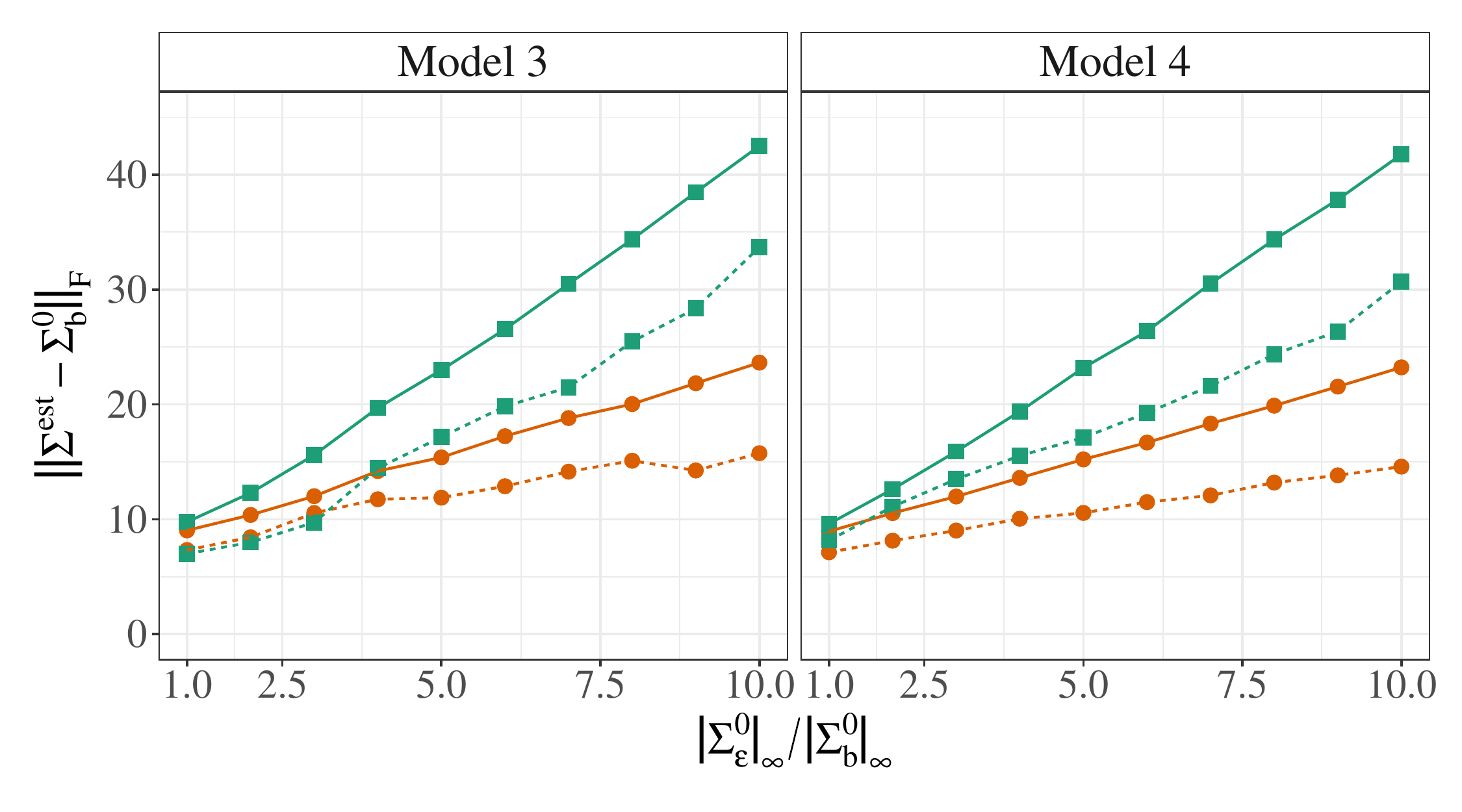}
% note that files may not be rotated
\caption{Estimation error (in Frobenius norm, averaged over 100 replicates) of the two between-subject sample covariance (solid) estimators ($\overline{\Sigma}$ and $\widehat{\Sigma}_b$) and their corresponding sparse and positive definite (dash) covariance estimators ($\overline{\Sigma}^+$ and $\widehat{\Sigma}_b^+$). The horizontal axis is the inverse signal-to-noise ratio, i.e., $|\Sigma_\veps^0|_\infty/|\Sigma_b^0|_\infty$. The estimation errors of $\overline{\Sigma}$ and $\overline{\Sigma}^+$ are marked in green, and the estimation errors of $\widehat{\Sigma}_b$ and $\widehat{\Sigma}_b^+$ are marked in orange.}
\label{fig:comparsion2}
\end{figure}

In general, our proposed between-subject sample estimate $\widehat{\Sigma}_{b}$ significantly outperforms $\overline{\Sigma}$ in both examples. This demonstrates the effect of the bias-correction as in \eqref{eq:naive_exp}. Moreover, for both sample estimators, their regularized versions (dash lines) achieve lower estimation errors, indicating the benefit of regularization. 

Surprisingly, as $|\Sigma_\veps^0|_\infty/|\Sigma_b^0|_\infty$ gets relatively small, $\overline{\Sigma}^+$ achieves an even smaller estimation error than $\widehat{\Sigma}_{b}^+$. This is an interesting cancellation of two biases with opposite signs: the estimation bias in the sample estimate $\overline{\Sigma}$ and the shrinkage bias in the $\ell_1$-penalty. 
Specifically, for any index pair $(j, k)$, \eqref{eq:naive_exp} indicates that the bias of $\overline{\Sigma}_{j, k}$ in estimating $(\Sigma^0_b)_{j, k}$ is $\sum_{i} (m n_i)^{-1} (\Sigma_\veps^0)_{j, k}$. In cases where $(\Sigma^0_\veps)_{j,k}$ and $(\Sigma^0_b)_{j, k}$ have the same signs (as in Model \ref{model3}), this sample estimation bias has the opposite effect from the shrinkage bias from the $\ell_1$ penalty. Consequently, these two biases could cancel each other when they have similar magnitudes, which is achieved when $(\Sigma^0_\veps)_{j, k}$ is on the similar scale as $\lambda$, and thus resulting in the surprisingly better performance of $\widehat{\Sigma}_b$ than $\widehat{\Sigma}_b^+$. Notably, when the estimation bias (as characterized by $|\Sigma_\veps^0|_\infty/|\Sigma_b^0|_\infty$) is too large to be canceled by the shrinkage bias, or when both biases have same signs (as in Model \ref{model4}), the performance of $\widehat{\Sigma}_b^+$ is dominating that of $\overline{\Sigma}^+$.

\section{Covariance Graphs of Clinical Variables from Hemodialysis Patients}
\label{sec:realdata}

We apply our proposed methods to estimate the between-subject and within-subject covariance structures among some clinical variables collected from hemodialysis patients. Hemodialysis is a treatment that filters wastes and fluid from patients' blood when the kidneys no longer function well. Hemodialysis patients usually follow a strict schedule by visiting a dialysis center about three times a week. Clinical variables, such as blood pressure and pulse, are measured during each treatment. Since numerous metabolic changes accompanying impaired kidney function affect all organ systems of the human body, it is imperative to study correlations among clinical variables. Those clinical variables are measured repeatedly for each hemodialysis patient at each treatment. We will investigate correlation structures at the patient (between-patient) and treatment (within-patient) levels. 

We use a dataset of measurements of several clinical and laboratory variables during 2018 and 2021 from 5,000 hemodialysis patients. For homogeneity, we consider white, non-diabetic, and non-Hispanic male patients who never had a COVID-19-positive polymerase chain reaction test. We use the measurements starting from the second year to avoid large fluctuations in the first year of dialysis. The dataset contains 276 patients with at least three complete treatment records every 30 days. The data imbalance is $\max_in_i / n_0 = 2.54$. For simplicity, we focus on the relationships among interdialytic weight gain, blood pressure, and heart rate. Based on \cite{Ipema.etal:2016}, we consider the following eight variables: \texttt{idwg} (interdialytic weight gain, kg), \texttt{ufv} (ultrafiltration volume, L),   \texttt{min\_sbp} (minimum systolic blood pressure, mmHg), \texttt{min\_dbp} (minimum diastolic blood pressure, mmHg), \texttt{max\_sbp} (maximum systolic blood pressure, mmHg), \texttt{max\_dbp} (maximum diastolic blood pressure, mmHg), \texttt{min\_pulse} (minimum pulse, beats/min), and \texttt{max\_pulse} (maximum pulse, beats/min). In our analysis, \texttt{ufv} is set to be the difference between predialysis and postdialysis weight within a hemodialysis session. 

We are interested in recovering the correlation structures at the patient and the treatment levels. Estimating the correlation matrix corresponds to recovering the correlation graph, where the nodes represent the random variables of interest and the edges present the marginal correlation between the nodes \citep{Chaudhuri.etal:2007}. We apply our method to repeated clinical measurements from these 276 patients. The regularization parameters are chosen by $5$-fold cross-validation with the one standard error rule \citep{Hastie.etal:2009}. Fig.~\ref{fig:graphs} presents estimates of the within-subject (left panel) and between-subject (middle panel) correlations, which indeed present different correlation structures. We also include the estimate using the aggregated data (right panel) for comparison, which coincides with our between-subject estimate. This is consistent with Theorem~\ref{thm:max.a.c} for this dataset's small value of $\max_in_i/n_0$.

\begin{figure}
\centering
\includegraphics[width=0.95\textwidth]{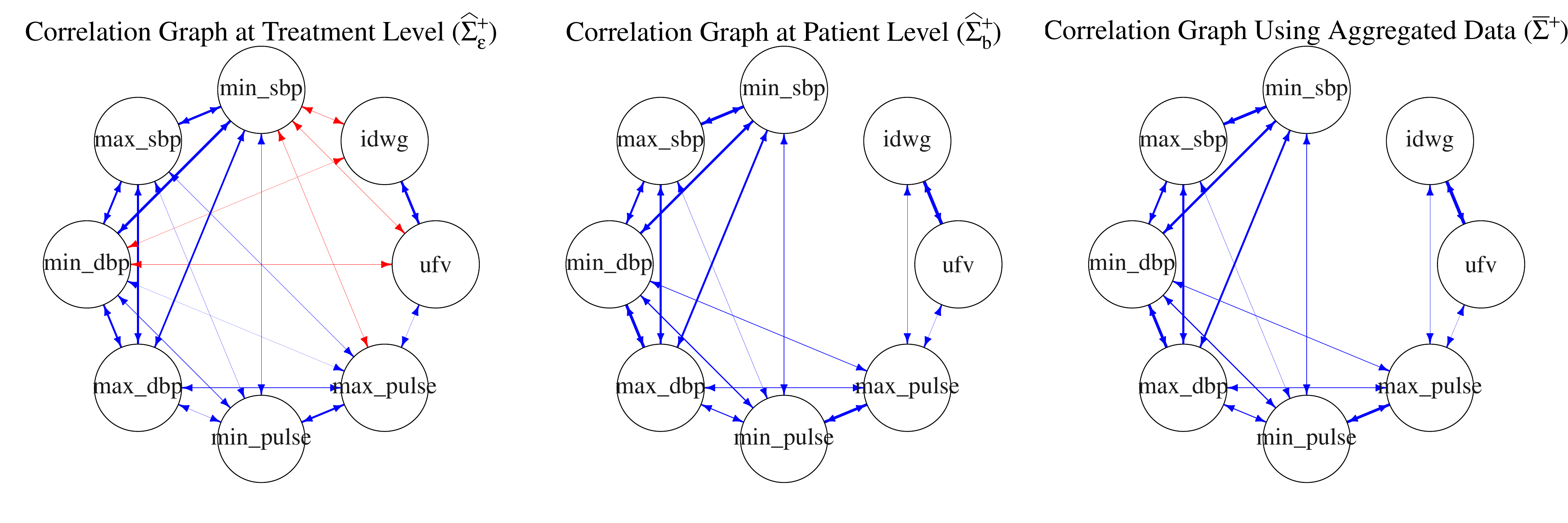}
% note that files may not be rotated
\caption{Within-subject (left), between-subject (middle) correlation graphs and correlation graph using the aggregated data (right) for clinical variables from hemodialysis patients. We present correlation matrices with the convention of using bi-directed covariance graphs \citep{Chaudhuri.etal:2007}. The blue edges correspond to the positive correlations, while the red edges represent the negative correlations. The width of an edge corresponds to the strength of the correlation.}\label{fig:graphs}
\end{figure}

It is important to realize that covariance structures at the treatment and patient levels could differ and should be estimated separately. Existing biological studies based on the aggregated measurements ignore such a difference in estimation and thus may lead to erroneous conclusions. In particular, our estimated correlation graph at the treatment level (within-subject) reveals much insight for hemodialysis treatment that cannot be recovered using the aggregate data. Specifically, we discuss several important recovered correlations in $\widehat{\Sigma}_\veps^+$ that have been missed in either $\widehat{\Sigma}_b^+$ or $\overline{\Sigma}^+$. Specifically, salt and fluid intake between two hemodialysis  sessions leads to interdialytic weight gain. A dialyzer, an artificial kidney, should filter the cumulation of waste and fluid. Ultrafiltration volume measures the waste and fluid removed from patients' blood. Consequently, higher \texttt{idwg} leads to larger \texttt{ufv}, confirmed by the positive correlation between \texttt{idwg} and \texttt{ufv} at the treatment level in Fig.~\ref{fig:graphs}. A rapid removal of fluid from a patient's blood results in the depletion of blood volume and subsequently leads to a decrease in systolic blood pressure, confirmed by the negative correlation between \texttt{ufv} and \texttt{min\_sbp} at the treatment level in Fig.~\ref{fig:graphs}. The lowered blood pressure will be compensated by heart functionality, which elevates the heart rate, again confirmed by the negative correlation between \texttt{min\_sbp} and \texttt{max\_pulse} at the treatment level in Fig.~\ref{fig:graphs}. However, no relationships among \texttt{idwg}, \texttt{max\_pulse} and \texttt{min\_sbp} have been observed at patient level in the middle panel of Fig.~\ref{fig:graphs}.
This implies that we should focus on correlations between clinical measurements at the treatment level rather than the patient level when evaluating the effectiveness of hemodialysis.

\section*{Acknowledgement}
This research was partially supported by NIH grant R01DK130067.
We thank Fresenius Medical Care North America for providing
de-identified data and Dr. Hanjie Zhang for discussing real data analysis.

\begin{appendices}
The appendices include further optimization implementation details, tuning parameter selection, comparison between the constrained and unconstrained estimators, and proofs of our main theorems.

\section{Further Details on Optimization Algorithm Implementation}
\label{sec:algo.details}

The complete algorithm solving the convex optimization problem \eqref{obj:admm0} is summarized in Algorithm \ref{alg:admm}.

\begin{algorithm}
\caption{Alternating direction method of multipliers for solving \eqref{obj:admm0}}\label{alg:admm}
\begin{algorithmic}[1]
\Require{$\delta$, $\lambda$, $\rho^{(0)}$, $B$, $\Sigma^{(0)}$,  $\Theta^{(0)}$, $\Lambda^{(0)}$, and $l = 0$.}
\State Repeat
\State $\qquad\Sigma^{(l+1)}\gets\frac{1}{1+\rho^{(l)}}\left(B+\rho \Theta^{(l)}-\Lambda^{(l)},\delta\right)_+$
\State $\qquad\Theta^{(l+1)}\gets\mathcal{S}_{\lambda/\rho^{(l)}}\left(\Sigma^{(l+1)}+\frac{1}{\rho^{(l)}}\Lambda^{(l)}\right)$
\State $\qquad\Lambda^{(l+1)}\gets\Lambda^{(l)}+\rho^{(l)}\left(\Sigma^{(l+1)}-\Theta^{(l+1)}\right)$
\State $\qquad$Update $\rho^{(l+1)}$ based on equation (3.13) in \citet{Boyd.etal:2010}
\State Until convergence
\end{algorithmic}
\end{algorithm}

A reasonable stopping criteria suggested by \citet{Boyd.etal:2010} is
\beqr
\|\Sigma^{(l+1)}-\Theta^{(l+1)}\|_F\le\eps^{\text{pri}}\qquad\text{and}\qquad\|\rho(\Theta^{(l+1)}-\Theta^{(l)})\|_F\le\eps^{\text{dual}}. \label{eq:stop1}
\eeqr
where $\eps^{\text{pri}}$ and $\eps^{\text{dual}}$ are positive feasibility tolerances for the primal and dual feasibility conditions, which are controlled by an absolute criterion $\eps^{\text{abs}}$ and a relative criterion $\eps^{\text{rel}}$: 
\beqr
\nonumber
\eps^{\text{pri}}&=&p\eps^{\text{abs}}+\eps^{\text{rel}}\max\{\|\Sigma^{(l+1)}\|_F,\|\Theta^{(l+1)}\|_F\},\\
\eps^{\text{dual}}&=&p\eps^{\text{abs}}+\eps^{\text{rel}}\|\Lambda^{(l+1)}\|_F, \label{eq:stop2}
\eeqr
where $\eps^{\text{abs}}>0$ and $\eps^{\text{rel}}>0$. In the numerical studies, we choose $\eps^{\text{abs}}=\eps^{\text{rel}}=10^{-8}$. Choice of $\rho$ can greatly impact the practical convergence of the alternating direction method procedure. And to improve the convergence, we adopt an adaptive strategy described in \citet{Boyd.etal:2010} for varying penalty parameter $\rho$. In practice, we use the soft-thresholding estimators based on the sample estimates as the initial $(\Sigma^{(0)}$,  $\Theta^{(0)})$. And the initial input for $\Lambda^{(0)}$ is a zero matrix. The initial penalty parameter $\rho$ is 0.1. Without the positive semi-definite constraints of $\Sigma_\veps$ and $\Sigma_b$ in \eqref{obj:admm0}, the unconstrained solutions will be $\mathcal{S}_{\lambda}(\widehat{\Sigma}_\veps)$ and $\mathcal{S}_{\lambda}(\widehat{\Sigma}_b)$ with $B=\widehat{\Sigma}_\veps$ and $B=\widehat{\Sigma}_b$, respectively. For efficient computation, we always first check the positive semi-definiteness of $\mathcal{S}_{\lambda}(\widehat{\Sigma}_\veps)$ and $\mathcal{S}_{\lambda}(\widehat{\Sigma}_b)$. If $\mathcal{S}_{\lambda}(\widehat{\Sigma}_\veps)$ and $\mathcal{S}_{\lambda}(\widehat{\Sigma}_b)$ are positive semi-definite, they are the final solutions to \eqref{obj:admm0}, respectively. Otherwise, we will use Algorithm \ref{alg:admm} to solve \eqref{obj:admm0}.

\section{Tuning parameters selection using cross-validation}
\label{sec:tunning}
The main optimization problem \eqref{obj:admm0} defines various estimators that we study in this paper, where $\lambda$ is the tuning parameter that controls the level of regularization of the sample estimates. We present in this section a cross-validation procedure for selecting the tuning parameter \citep{Bickel:Levina:2008b,Rothman.etal:2009,Cai:Liu:2011} specifically in the presence of repeated measurements. 

For each (of the $K$) split in a K-fold cross-validation procedure, we randomly partition the $m$ groups into a set of $m_1$ groups of training set, i.e., $\mathcal{T}_r=\{\bY_{ij}:i\in\mathcal{A}\}$ with $|\mathcal{A}|=m_1$ and a set of $m-m_1$ groups of validation set, i.e.,  $\mathcal{T}_e=\{\bY_{ij}:i\in\mathcal{A}^c\}$ with $|\mathcal{A}^c|=m-m_1$.

Let $\widehat{S}^+\{\lambda,\widehat{S}(\mathcal{T})\}$ denote a generic estimator, which is defined as a solution to the optimization problem (5) in the main text with the tuning parameter value $\lambda$ and input sample matrix $\widehat{S}(\mathcal{T})$ evaluated using a dataset $\mathcal{T}$. Specifically, the estimator $\widehat{S}^+\{\lambda, \widehat{S}(\mathcal{T})\}$ could refer to $\widehat{\Sigma}_b^+$, $\widehat{\Sigma}_\varepsilon^+$, $\widetilde{\Sigma}_b^+$, and $\overline{\Sigma}^+$. And $\widehat{S} (\mathcal{T})$ refers to the unbiased estimator $\widehat{\Sigma}_b$, $\widehat{\Sigma}_\varepsilon$, $\widetilde{\Sigma}_b$, and the biased estimator $\overline{\Sigma}$.
The cross-validation procedure is presented in the following Algorithm \ref{alg:cv} to choose the tuning parameter from a path of candidate tuning parameter values $\{ \lambda_1 > \lambda_2  > \ldots > \lambda_L \}$.

\begin{algorithm}
\caption{a K-fold Cross-Validation Procedure}\label{alg:cv}
\begin{algorithmic}[1]
\Require{$\{\bY_{ij}: 1\le i\le m, 1\le j\le n_i\}$ and $\{ \lambda_1 > \lambda_2  > \ldots > \lambda_L \}$.}
\For{$\ell = 1, \ldots, L$}
\For{$\nu=1,\ldots, K$}
\State Divide $\{\bY_{ij}: 1\le i\le m, 1\le j\le n_i\}$ into training set $\mathcal{T}^{(\nu)}_r$ and validation set  $\mathcal{T}^{(\nu)}_e$;
\State Compute the sample covariance matrix $\widehat{S}(\mathcal{T}_e^{(\nu)})$ on the validation set $\mathcal{T}^{(\nu)}_e$;
\State Compute the estimator $\widehat{S}^+\{\lambda_\ell,\widehat{S}(\mathcal{T}_r^{(\nu)})\}$ on the training set $\mathcal{T}^{(\nu)}_r$.
\EndFor
\State Compute CV estimate of error $E_\ell = \sum_{\nu=1}^K \|\widehat{S}^+\{\lambda_\ell,\widehat{S}(\mathcal{T}_r^{(\nu)})\}-\widehat{S}(\mathcal{T}_e^{(\nu)})\|_F^2 / K$.
\EndFor
\State Let $\hat{\ell} = \argmin_{\ell = 1, \ldots, L} E_\ell$, and return the selected tuning parameter $\lambda_{\hat{\ell}}$.
\end{algorithmic}
\end{algorithm}

\section{Lemmas and Proofs of the Main Theorems}
\label{sec:proofs}

\subsection{Lemmas}
\label{sec:lem}

We observe $\bY_{ij} \in \mathbb{R}^p$, which is the $j$-th repeated measurement of the $i$-th subject for $j = 1,\ldots, n_i$ and $i = 1,\ldots,m$, following the model \eqref{eq:simpleY}, where $\bm{\veps}_{ij}$ and $\bm{b}_i$ are $p$-dimensional sub-Gaussian random vectors with the true within-subject and between-subject covariance, i.e., $\mathrm{cov}(\bm{\veps}_{ij}) = \Sigma_\veps^0$ and $\mathrm{cov}(\bm{b}_i) = \Sigma_b^0$, respectively, and $\bm{b}_i$ and $\bm{\veps}_{ij}$ are mutually independent. Let $N = \sum_{i = 1}^m n_i$ be the total number of observations. We begin with several lemmas which are essential for the proofs of the main results.

\begin{lemma}\label{lem:max.r}
Consider the true within-subject covariance $\Sigma^0_{\veps}$ with  $\max_k (\Sigma^0_\veps)_{k,k}\le M_{\veps}$. Let $\lambda_\veps =C_1\{N\log p\}^{1/2}/(N-m)$ for a sufficiently large constant $C_1$. If $\log p\le N$, then the unbiased within-subject sample estimate $\widehat{\Sigma}_\veps$ satisfies
\beqrs
\text{pr}\left\{\max_{k,l}\left|(\widehat{\Sigma}_{\veps}-\Sigma^0_{\veps})_{k,l}\right|>\lambda_\veps\right\}\le4p^{-C_2}, \label{eq:lem.max.r}
\eeqrs
where $C_2>0$ only depends on $C_1$ and $M_\veps$.
\end{lemma}

\begin{proof}
We first rewrite $\widehat{\Sigma}_\veps$ as follows, 
\beqrs
\nonumber
\widehat{\Sigma}_\veps&=&\frac{1}{N-m}\sum_{i=1}^m\sum_{j=1}^{n_i}(\bY_{ij}-\bar{\bY}_{i\cdot})(\bY_{ij}-\bar{\bY}_{i\cdot})^\T\\
&=&\frac{1}{N-m}\left(\sum_{i=1}^m\sum_{j=1}^{n_i}\bm{\veps}_{ij}\bm{\veps}_{ij}^\T-\sum_{i=1}^mn_i\bar{\bm{\veps}}_{i\cdot}\bar{\bm{\veps}}_{i\cdot}^\T\right). \label{eq:s.m.mdR2.0}
\eeqrs
Then,
\beqr
\nonumber
(\widehat{\Sigma}_{\veps})_{k,l}&=&\frac{1}{N-m}\left(\sum_{i=1}^m\sum_{j=1}^{n_i}\veps_{ijk}\veps_{ijl}-\sum_{i=1}^mn_i\bar{\veps}_{i\cdot k}\bar{\veps}_{i\cdot l}\right)\\ \nonumber
&=&\frac{1}{N-m}\sum_{i=1}^m\sum_{j=1}^{n_i}\veps_{ijk}\veps_{ijl}-\frac{1}{N-m}\sum_{i=1}^mn_i\bar{\veps}_{i\cdot k}\bar{\veps}_{i\cdot l}\\ \nonumber
&=&\frac{1}{N-m}\sum_{i=1}^m\sum_{j=1}^{n_i}\left\{\veps_{ijk}\veps_{ijl}-(\Sigma^0_{\veps})_{k,l}\right\}\\
&&-\frac{1}{N-m}\sum_{i=1}^m\left\{\frac{1}{n_i}S_{i\cdot k}S_{i\cdot l}-(\Sigma^0_{\veps})_{k,l}\right\}+(\Sigma^0_{\veps})_{k,l}, \label{eq:s.m.mdR2.1}
\eeqr
where $S_{i\cdot k}=\sum_{j=1}^{n_i}\veps_{ijk}$.

By (\ref{eq:s.m.mdR2.1}),
\beqr
\nonumber
&&\max_{k,l}\left|(\widehat{\Sigma}_{\veps}-\Sigma^0_{\veps})_{k,l}\right|\\ \nonumber
&\le&\max_{k,l}\frac{1}{N-m}\left|\sum_{i=1}^m\sum_{j=1}^{n_i}\left\{\veps_{ijk}\veps_{ijl}-(\Sigma^0_{\veps})_{k,l}\right\}\right|\\ 
&&+\frac{1}{N-m}\max_{k,l}\left|\sum_{i=1}^m\left\{\frac{1}{n_i}S_{i\cdot k}S_{i\cdot l}-(\Sigma^0_{\veps})_{k,l}\right\}\right|. \label{eq:s.m.mdR2.2}
\eeqr

Now, we assume that $\veps_{ijk}\in\mathcal{SG}(\sigma^2_{\veps,k})$, i.e., $\veps_{ijk}$ is sub-Gaussian with a variance factor $\sigma_{\veps,k}^2$ for $1\le i\le m, 1\le j\le n_i, 1\le k\le p$. 
%And a random variable $X$ with mean $\mu$ is sub-Gaussian with a variance factor $\sigma^2$, if $E[\text{exp}\{t(X-\mu)\}]\le\text{exp}(t^2\sigma^2/2)$ for all $t\in\mathbb{R}$.
It is easy to check that $n_i^{-1/2}S_{i\cdot k}\in\mathcal{SG}(\sigma_{\veps,k}^2)$.

Let $\psi:\mathbb{R}_+\to\mathbb{R}_+$ be a convex function with $\psi(0)=0$, especially, $\psi_q(v)= \text{exp}(|v|^q)-1$, for $q\in[1,2]$. Then for an $\mathbb{R}$-valued random variable $X$, the Orlicz norm of $X$ is $\|X\|_{\psi}= \inf\{t\in\mathbb{R}_+: E\{\psi(|X|/t)\}\le1\}.$ And by the properties of Orlicz norms, for any random variable $X$ and any increasing convex $\psi:\mathbb{R}_+\to\mathbb{R}_+$ with $\psi(0)=0$, we have
\beqr
\|X-E(X)\|_{\psi}\le2\|X\|_{\psi}. \label{eq:subGp1}
\eeqr
Moreover, if $X\in\mathcal{SG}(\sigma^2)$, then 
\beqr
\|X\|_{\psi_2}\le c_0\sigma,\label{eq:subGp2}
\eeqr
for some $c_0\le(8/3)^{1/2}$.

Since $\veps_{ijk}\in\mathcal{SG}(\sigma^2_{\veps,k})$ and $n_i^{-1/2}S_{i \cdot k}\in\mathcal{SG}(\sigma_{\veps,k}^2)$, by Lemma 2.7.7 in \citet{Vershyin:2018}, $\veps_{ijk}\veps_{ijl}$ and $n_i^{-1}S_{i \cdot k} S_{i \cdot l}$  are sub-Exponential random variables. 
%A random variable $X$ with mean $\mu$ is sub-Exponential ($\mathcal{SE}$) with parameters $\nu^2$ and $\alpha$, where $\nu,\alpha>0$, if $E[\text{exp}\{t(X-\mu)\}]\le\text{exp}(t^2\nu^2/2)$ for $|t|\le1/\alpha$.
Let $\max_k \sigma_{\veps,k}^2=M_\veps$. Combining (\ref{eq:subGp1}) and (\ref{eq:subGp2}), Lemma 2.7.7 in \citet{Vershyin:2018} implies that
\beqrs
\left\|\veps_{ijk}\veps_{ijl}-(\Sigma^0_{\veps})_{k,l}\right\|_{\psi_1}\le2\left\|\veps_{ijk}\veps_{ijl}\right\|_{\psi_1} \le 2\left\|\veps_{ijk}\right\|_{\psi_2}\left\|\veps_{ijl}\right\|_{\psi_2}\le c_1M_{\veps},\label{eq:subGpR1.01}
\eeqrs
and
\beqrs
\left\|n_i^{-1}S_{i\cdot k}S_{i\cdot l} -(\Sigma^0_{\veps})_{k,l}\right\|_{\psi_1}\le2\left\|n_i^{-1}S_{i\cdot k}S_{i\cdot l}\right\|_{\psi_1} 
\le 2\left\|n_i^{-1/2}S_{i\cdot k} \right\|_{\psi_2}\left\|n_i^{-1/2}S_{i\cdot l}\right\|_{\psi_2}\le c_1M_{\veps},\label{eq:subGpR1.02}
\eeqrs
where $c_1=2c_0^2$. 

Hence, for the first term in (\ref{eq:s.m.mdR2.2}), by the union sum inequality and Bernstein's inequality (Theorem 2.8.2 in \citet{Vershyin:2018}), we can get
\beqr
\nonumber
&&\text{pr}\left[\max_{k,l}\frac{1}{N-m}\left|\sum_{i=1}^m\sum_{j=1}^{n_i}\{\veps_{ijk}\veps_{ijl}-(\Sigma^0_{\veps})_{k,l}\}\right|\ge t\right]\\
&\le& 2p^2\text{exp}\left[-c_2\min\left\{\frac{t^2(N-m)^2}{NK_1^2},\frac{t(N-m)}{K_1}\right\}\right], \label{eq:s.m.pbR1.3}
\eeqr
where $c_2>0$, $K_1=\max_{i,k,l}\|\veps_{ijk}\veps_{ijl}-(\Sigma^0_{\veps})_{k,l}\|_{\psi_1}\le c_1M_\veps$.

Similarly, 
\beqr
\nonumber
&&\text{pr}\left[\frac{1}{N-m}\max_{k_1}\left|\sum_{i=1}^m\left\{\frac{1}{n_i}S_{i\cdot k_1}S_{i\cdot k_2}-(\Sigma^0_{\veps})_{k,l}\right\}\right|\ge t\right]\\
&\le& 2p^2\text{exp}\left[-c_3\min\left\{\frac{t^2(N-m)^2}{mK_2^2},\frac{t(N-m)}{K_2}\right\}\right], \label{eq:s.m.pbR1.2}
\eeqr
where $c_3>0$, $K_2=\max_{i,k,l}\|n_i^{-1}S_{i\cdot k}S_{i\cdot l}-(\Sigma^0_{\veps})_{k,l}\|_{\psi_1}\le c_1M_\veps$.

By \eqref{eq:s.m.pbR1.3} and \eqref{eq:s.m.pbR1.2}, take $t = C_1(N\log p)^{1/2}/\{2(N-m)\}$ for a sufficiently large constant $C_1>0$, with $N>\log p$, we will have
\beqr
\nonumber
&&\text{pr}\left[\max_{k,l}\frac{1}{N-m}\left|\sum_{i=1}^m\sum_{j=1}^{n_i}\{\veps_{ijk}\veps_{ijl}-(\Sigma^0_{\veps})_{k,l}\}\right|\ge t\right]\\ \nonumber
&\le&2\text{exp}\left[\max\left\{\left(2-\frac{c_2NC_1^2}{4mK_1^2}\right)\log p,2\log p-\frac{c_2C_1}{2K_1}(N\log p)^{1/2}\right\}\right]\\
&\le&2\text{exp}\left\{\max\left(2-\frac{c_2C_1^2}{4c_1^2M_{\veps}^2},2-\frac{c_2C_1}{2c_1M_{\veps}}\right)\log p\right\},\label{eq:s.m.pbR1.5}
\eeqr
and
\beqr
\nonumber
&&\text{pr}\left[\frac{1}{N-m}\max_{k,l}\left|\sum_{i=1}^m\left\{\frac{1}{n_i}S_{i\cdot k}S_{i\cdot l}-(\Sigma^0_{\veps})_{k,l}\right\}\right|\ge t\right]\\ \nonumber
&\le&2\text{exp}\left[\max\left\{\left(2-\frac{c_3NC_1^2}{4mK_2^2}\right)\log p,2\log p-\frac{c_3C_1}{2K_2}(N\log p)^{1/2}\right\}\right]\\ 
&\le&2\text{exp}\left\{\max\left(2-\frac{c_3C_1^2}{4c_1^2M_{\veps}^2},2-\frac{c_3C_1}{2c_1M_{\veps}}\right)\log p\right\}. \label{eq:s.m.pbR1.4}
\eeqr

Combining (\ref{eq:s.m.pbR1.5}) and (\ref{eq:s.m.pbR1.4}), with $\lambda_\veps=C_1(N\log p)^{1/2}/(N-m)$, we have
\beqrs
\nonumber
&&\text{pr}\left\{\max_{k,l}\left|(\widehat{\Sigma}_\veps-\Sigma_\veps^0)_{k,l}\right|>\lambda_\veps\right\}\\ \nonumber
&\le&\text{pr}\left[\frac{1}{N-m}\max_{k_1}\left|\sum_{i=1}^m\left\{\frac{1}{n_i}S_{i\cdot k}S_{i\cdot l}-(\Sigma^0_{\veps})_{k,l}\right\}\right|\ge \frac{C_1(N\log p)^{1/2}}{2(N-m)}\right]\\ \nonumber
&&\qquad+\text{pr}\left[\max_{k,l}\frac{1}{N-m}\left|\sum_{i=1}^m\sum_{j=1}^{n_i}\{\veps_{ijk}\veps_{ijl}-(\Sigma^0_{\veps})_{k,l}\}\right|\ge \frac{C_1(N\log p)^{1/2}}{2(N-m)}\right]\\ \nonumber
&\le&2\text{exp}\left\{\max\left(2-\frac{c_3C_1^2}{4c_1^2M_{\veps}^2},2-\frac{c_3C_1}{2c_1M_{\veps}}\right)\log p\right\}\\ \nonumber
&&\qquad+2\text{exp}\left\{\max\left(2-\frac{c_2C_1^2}{4c_1^2M_{\veps}^2},2-\frac{c_2C_1}{2c_1M_{\veps}}\right)\log p\right\}\\
&\le&4p^{-C_2}, \label{eq:s.m.pbR1.6}
\eeqrs
where $C_2=\min\{c_3C_1(2c_1M_{\veps})^{-1}, c_3,c_2C_1(2c_1M_{\veps})^{-1},c_2\}(2c_1M_{\veps})^{-1}C_1-2$.
\end{proof}

\begin{lemma}\label{lem:max.b2}
Consider the true within-subject covariance $\Sigma^0_{\veps}$ with  $\max_k (\Sigma^0_\veps)_{k,k}\le M_{\veps}$ and the true between-subject covariance $\Sigma^0_b$ with  $\max_k (\Sigma^0_b)_{k,k}\le M_b$. Let 
\beqrs
\lambda_b = C_1 \left(\frac{\log p}{m}\right)^{1/2} + C_2 \frac{\left( N\log p\right)^{1/2}}{(N - m) n^\ast} + \frac{M_b}{m} + \frac{M_\veps}{m n^\ast } 
\eeqrs
for sufficiently large $C_1, C_2>0$, where $n^\ast=m/\sum_{i=1}^mn_i^{-1}$. If $\log p\le m$, then the unbiased between-subject sample estimate $\widehat{\Sigma}_b$ satisfies 
\beqrs
\text{pr}\left\{\max_{k,l}\left|(\widehat{\Sigma}_b-\Sigma_b^0)_{k,l}\right|>2\lambda_b\right\}\le8p^{-C_3}, \label{eq:lem.max.b2}
\eeqrs
where $C_3> 0$ only depends on $C_1$, $C_2$ and $\max(M_\veps,M_b)$.
\end{lemma}

\begin{proof}
%Now, we will  consider the convergence rate of $\max_{k,l}|(\widehat{\Sigma}_b-\Sigma^0_b)_{k,l}|$. 
Let $\bar{Y}_{i\cdot k}=b_{ik}+n_i^{-1}\sum_{j=1}^{n_i}\veps_{ijk}=b_{ik}+n_i^{-1}S_{i\cdot k}=W_{ik}$, then by decomposition,  
\beqr
\nonumber
(\widehat{\Sigma}_b-\Sigma^0_b)_{k,l}&=&\{\overline{\Sigma}-(n^\ast)^{-1}\widehat{\Sigma}_\veps-\Sigma^0_b\}_{k,l}\\ \nonumber
&=&[\overline{\Sigma}-\{\Sigma^0_b+(n^\ast)^{-1}\Sigma^0_\veps\}]_{k,l}-(n^\ast)^{-1}(\widehat{\Sigma}_\veps-\Sigma^0_\veps)_{k,l}\\ \nonumber
&=&\frac{1}{m-1}\sum_{i=1}^m\left\{W_{ik}W_{il}-\left(\Sigma_b^0+n_i^{-1}\Sigma^0_{\veps}\right)_{k,l}\right\}\\ \nonumber
&&-\frac{m}{m-1}\left(\frac{1}{m}\sum_{i=1}^mW_{ik}\right)\left(\frac{1}{m}\sum_{i=1}^mW_{il}\right)\\
&&+\frac{(\Sigma_b^0)_{k,l}}{m-1}+\frac{(\Sigma^0_{\veps})_{k,l}}{(m-1)n^\ast}-\frac{(\widehat{\Sigma}_\veps-\Sigma^0_\veps)_{k,l}}{n^\ast}. \label{eq:b2.dc.1}
\eeqr

Then, with $|(\Sigma_b^0)_{k,l}|\le M_b$ and $|(\Sigma^0_{\veps})_{k,l}|\le M_{\veps}$, we have
\beqr
\nonumber
\max_{k,l}\left|(\widehat{\Sigma}_b-\Sigma^0_b)_{k,l}\right|&\le&2\max_{k,l}\left|\frac{1}{m}\sum_{i=1}^m\left\{W_{ik}W_{il}-\left(\Sigma_b^0+n_i^{-1}\Sigma^0_{\veps}\right)_{k,l}\right\}\right|\\ \nonumber
&&+2\max_{k,l}\left|\left(\frac{1}{m}\sum_{i=1}^mW_{ik}\right)\left(\frac{1}{m}\sum_{i=1}^mW_{il}\right)\right|\\ \nonumber
&&+\max_{k,l}(n^\ast)^{-1}\left|(\widehat{\Sigma}_\veps-\Sigma^0_\veps)_{k,l}\right|\\
&&+\frac{2M_b}{m}+\frac{2M_\veps}{mn^\ast}. \label{eq:b2.dc.2}
\eeqr

Assume that $b_{ik}\in\mathcal{SG}(\sigma_{b,k}^2)$, i.e., $b_{ik}$ is sub-Gaussian with a variance factor $\sigma_{b,k}^2$ for $1\le i\le m, 1\le k\le p$. Then $W_{ik}\in\mathcal{SG}(\sigma_{b,k}^2+n_i^{-1}\sigma_{\veps,k}^2)$. Let $\max_k \sigma_{b,k}^2=M_b$. Then, by Lemma 2.7.7 in \citet{Vershyin:2018}, we obtain
\beqrs
\left\|W_{ik}W_{il}-\left(\Sigma_b^0+n_i^{-1}\Sigma^0_\veps\right)_{k,l}\right\|_{\psi_1}&\le&2\left\|W_{ik}\right\|_{\psi_2}\left\|W_{il}\right\|_{\psi_2}\\
&\le& c_1\left(\Sigma_b^0+n_i^{-1}\Sigma^0_\veps\right)_{k,l}\\
&\le& c_1\left(1+n_l^{-1}\right)M_\ast\\
&\le& 2c_1M_\ast,
\eeqrs
where $n_l=\min n_i$ and $M_\ast=\max (M_\veps,M_b)$. And with the Bernstein's inequality, we have
\beqrs
&&\text{pr}\left[\frac{1}{m}\left|\sum_{i=1}^m\left\{W_{ik}W_{il}-(\Sigma_b^0+n_i^{-1}\Sigma_\veps^0)_{k,l}\right\}\right|\ge t\right]\le 2\text{exp}\left\{-c_4\min\left(\frac{mt^2}{K_3^2},\frac{mt}{K_3}\right)\right\}, \label{eq:b2.prob1}
\eeqrs
where $c_4>0$, $K_3=\max_{i,k,l}\|W_{ik}W_{il}-(\Sigma_b^0+n_i^{-1}\Sigma_\veps^0)_{k,l}\|_{\psi_1}\le 2c_1M_\ast$.

By the union sum inequality and taking $t=2^{-1}C_1(\log p/m)^{1/2}$ for a sufficiently large constant $C_1>0$, if $m\ge\log p$, we have
\beqr
\nonumber
&&\text{pr}\left[\max_{k,l}\frac{1}{m}\left|\sum_{i=1}^m\left\{W_{ik}W_{il}-(\Sigma_b^0+n_i^{-1}\Sigma^0_\veps)_{k,l}\right\}\right|\ge t\right]\\ \nonumber
&\le& 2p^2\text{exp}\left[-c_4\min\left\{\frac{C_1^2\log p}{4K_3^2},\frac{C_1(m\log p)^{1/2}}{2K_3}\right\}\right]\\
&\le&2\text{exp}\left[\left\{2-\min\left(\frac{c_4C_1^2}{16c_1^2M_\ast^2},\frac{c_4C_1}{4c_1M_\ast}\right)\right\}\log p\right]. \label{eq:b2.prob2}
\eeqr

We use a union bound with the general Hoeffding's inequality (Theorem 2.6.2 by \citet{Vershyin:2018}) to bound the second term in \eqref{eq:b2.dc.2}. Specifically, with $m\ge\log p$ and taking $t=2^{-1}C_1(\log p/m)^{1/2}$, we have
\beqr
\nonumber
\text{pr}\left(\max_{k,l}\left|\frac{1}{m}\sum_{i=1}^mW_{ik}\right|^2\ge t\right)
&=&\text{pr}\left(\max_{k,l}\left|\sum_{i=1}^mW_{ik}\right|\ge mt^{1/2}\right)\\ \nonumber
&\le&2p\text{exp}\left(-\frac{c_5m^2t}{\sum_{i=1}^m\|W_{ik}\|_{\psi_2}^2}\right)\\ \nonumber
&\le&2p\text{exp}\left(-\frac{c_5mt}{c_1M_\ast}\right)\\ \nonumber
&=&2p\text{exp}\left\{-\frac{c_5C_1}{2c_1M_\ast}(m\log p)^{1/2}\right\}\\
&\le&2\text{exp}\left\{\left(1-\frac{c_5C_1}{2c_1M_\ast}\right)\log p\right\}, \label{eq:b2.prob3}
\eeqr
where $c_5>0$.

For the third term in \eqref{eq:b2.dc.2}, by Lemma \ref{lem:max.r}, for a sufficiently large constant $C_2>0$, we have
\beqr
\text{pr}\left\{\max_{k,l}\frac{\left|(\widehat{\Sigma}_\veps-\Sigma_\veps^0)_{k,l}\right|}{n^\ast}\ge 2C_2\frac{(N\log p)^{1/2}}{(N-m)n^\ast}\right\}\le 4p^{-C_3^\prime}, \label{eq:b2.prob4}
\eeqr
where $C_3^\prime>0$ only depends on $C_2$ and $M_\veps$.

Collecting \eqref{eq:b2.prob2}-\eqref{eq:b2.prob4}, with
\beqrs
\lambda_b=C_1\left(\frac{\log p}{m}\right)^{1/2}+C_2\frac{(N\log p)^{1/2}}{(N-m)n_*}+\frac{M_b}{m}+\frac{M_\veps}{mn_*},
\eeqrs 
we have
\beqrs
\nonumber
&&\text{pr}\left\{\max_{k,l}\left|(\widehat{\Sigma}_b-\Sigma_b^0)_{k,l}\right|\ge2\lambda_b\right\}\\ \nonumber
&\le&\text{pr}\left[\max_{k,l}\frac{2}{m}\left|\sum_{i=1}^m\left\{W_{ik}W_{il}-(\Sigma_b^0+n_i^{-1}\Sigma_\veps^0)_{k,l}\right\}\right|\ge C_1\left(\frac{\log p}{m}\right)^{1/2}\right]\\ \nonumber
&&\qquad+\text{pr}\left\{2\max_{k,l}\left|\left(\frac{1}{m}\sum_{i=1}^mW_{ik}\right)\left(\frac{1}{m}\sum_{i=1}^mW_{il}\right)\right|\ge C_1\left(\frac{\log p}{m}\right)^{1/2}\right\}\\ \nonumber
&&\qquad+\text{pr}\left\{\max_{k,l}\frac{\left|(\widehat{\Sigma}_\veps-\Sigma_\veps^0)_{k,l}\right|}{n^\ast}\ge 2C_2\frac{(N\log p)^{1/2}}{(N-m)n^\ast}\right\}\\ \nonumber
&\le&4p^{-C_3^\prime}+4p^{-C_3^{\dprime}}\\
&\le&8p^{-C_3}, \label{eq:b2.prob5}
\eeqrs
where $C_3^{\dprime} = \min\{c_4C_1^2(16c_1^2M_\ast^2)^{-1},c_4C_1(4c_1M_\ast)^{-1},c_5C_1(2c_1M_\ast)^{-1}+1\}-2$ and $C_3 = \min(C_3^\prime,C_3^{\dprime})$.
\end{proof}

%%%%%%
%%%%%%

\begin{lemma}\label{lem:max.a}
Consider the true within-subject covariance $\Sigma^0_{\veps}$ with  $\max_k (\Sigma^0_\veps)_{k,k}\le M_{\veps}$ and the true between-subject covariance $\Sigma^0_b$ with  $\max_k (\Sigma^0_b)_{k,k}\le M_b$. Let 
\beqrs
\lambda_0 = C_1 \left(\frac{\log p}{m}\right)^{1/2} + \frac{M_b}{m} + \frac{M_\veps}{n^\ast } 
\eeqrs
for sufficiently large $C_1>0$, where $n^\ast=m/\sum_{i=1}^mn_i^{-1}$. If $\log p\le m$, then the naive between-subject sample estimate $\overline{\Sigma}$ satisfies 
\beqrs
\text{pr}\left\{\max_{k,l}\left|(\overline{\Sigma}-\Sigma_b^0)_{k,l}\right|>2\lambda_b\right\}\le8p^{-C_2}
\eeqrs
where $C_2> 0$ only depends on $C_1$ and $\max(M_\veps,M_b)$.
\end{lemma}

\begin{proof}
Now, we will  consider the convergence rate of $\max_{k,l}|(\overline{\Sigma}-\Sigma^0_b)_{k,l}|$. By \eqref{eq:b2.dc.1}, we have
\beqr
\nonumber
(\overline{\Sigma}-\Sigma^0_b)_{k,l}&=&\frac{1}{m-1}\sum_{i=1}^m\left\{W_{ik}W_{il}-\left(\Sigma_b^0+n_i^{-1}\Sigma^0_{\veps}\right)_{k,l}\right\}\\ \nonumber
&&-\frac{m}{m-1}\left(\frac{1}{m}\sum_{i=1}^mW_{ik}\right)\left(\frac{1}{m}\sum_{i=1}^mW_{il}\right)\\
&&+\frac{(\Sigma_b^0)_{k,l}}{m-1}+\frac{m(\Sigma^0_{\veps})_{k,l}}{(m-1)n^\ast}. \label{eq:a.dc.1}
\eeqr

Then, with $|(\Sigma_b^0)_{k,l}|\le M_b$ and $|(\Sigma^0_{\veps})_{k,l}|\le M_{\veps}$, we have
\beqr
\nonumber
\max_{k,l}\left|(\overline{\Sigma}-\Sigma^0_b)_{k,l}\right|&\le&2\max_{k,l}\left|\frac{1}{m}\sum_{i=1}^m\left\{W_{ik}W_{il}-\left(\Sigma_b^0+n_i^{-1}\Sigma^0_{\veps}\right)_{k,l}\right\}\right|\\ \nonumber
&&+2\max_{k,l}\left|\left(\frac{1}{m}\sum_{i=1}^mW_{ik}\right)\left(\frac{1}{m}\sum_{i=1}^mW_{il}\right)\right|\\
&&+\frac{2M_b}{m}+\frac{2M_\veps}{n^\ast}. \label{eq:a.dc.2}
\eeqr

Following the steps in Lemma \ref{lem:max.b2}, with
\beqrs
\lambda_0=C_1\left(\frac{\log p}{m}\right)^{1/2}+\frac{M_b}{m}+\frac{M_\veps}{n^\ast}
\eeqrs
for a sufficiently large constant $C_1>0$, we have
\beqrs
\text{pr}\left\{\max_{k,l}\left|(\overline{\Sigma}-\Sigma_b^0)_{k,l}\right|>2\lambda_0\right\}\le4p^{-C_2},
\eeqrs
where $C_2>0$ only depends on $C_1$ and $\max(M_\veps,M_b)$.
\end{proof}

\begin{lemma}\label{lem:max.g}
Consider the true within-subject covariance $\Sigma^0_{\veps}$ with  $\max_k (\Sigma^0_\veps)_{k,k}\le M_{\veps}$ and the true between-subject covariance $\Sigma^0_b$ with  $\max_k (\Sigma^0_b)_{k,k}\le M_b$. Let
$$
\widetilde{\lambda}_b = C_1 \frac{\max_i n_i}{n_0} \left(\frac{\log p}{m}\right)^{1/2} + C_2 \frac{\left( N\log p\right)^{1/2}}{n_0 (N - m)}
+ \frac{(2N-n_0m)M_b}{2n_0 m} + \frac{M_\veps}{n_0 m}
$$
for sufficiently large $C_1, C_2 > 0$.
If $\log p\le m$, then $\widetilde{\Sigma}_b$  satisfies 
\beqrs
\text{pr}\left[\max_{k,l}\left|(\widetilde{\Sigma}_b - \Sigma_b^0)_{k,l}\right|> 2 \widetilde{\lambda}_b\right]\le8p^{-C_3}, 
\eeqrs
where $C_3> 0$ only depends on $C_1$, $C_2$ and $\max(M_\veps,M_b)$. 
\end{lemma}

\begin{proof}
%Now, we will deal with the convergence rate for $\widetilde{\Sigma}_b^+$:
% and first consider the convergence rate of $\max_{k,l}|(\widetilde{\Sigma}_b-\Sigma_b^0)_{k,l}|$,
Consider
\beqr
\nonumber
\max_{k,l}|(\widetilde{\Sigma}_b-\Sigma_b^0)_{k,l}|&=&\max_{k,l}\left|\left(\frac{\overline{\Sigma}-\widehat{\Sigma}_\veps}{n_0}-\Sigma_b^0\right)_{k,l}\right|\\ \nonumber
&=&\max_{k,l}\left|\frac{(\overline{\Sigma}-n_0\Sigma_b^0-\Sigma_\veps^0)_{k,l}}{n_0}-\frac{(\widehat{\Sigma}_\veps-\Sigma_\veps^0)_{k,l}}{n_0}\right|\\
&\le&\max_{k,l}\left|\frac{(\overline{\Sigma}-n_0\Sigma_b^0-\Sigma_\veps^0)_{k,l}}{n_0}\right|+\max_{k,l}\left|\frac{(\widehat{\Sigma}_\veps-\Sigma_\veps^0)_{k,l}}{n_0}\right|. \label{eq:s.m.mdG1.0}
\eeqr

With $\bar{Y}_{i\cdot k}=W_{ik}$, we have
\beqrs
\bar{Y}_{\cdot\cdot k}&=&\frac{1}{N}\sum_{i=1}^mn_iW_{ik},\\
(\overline{\Sigma})_{k,l}&=&\frac{1}{m-1}\sum_{i=1}^mn_i\left(W_{ik}-\frac{1}{N}\sum_{i=1}^mn_iW_{ik}\right)\left(W_{il}-\frac{1}{N}\sum_{i=1}^mn_iW_{il}\right)\\
&=&\frac{1}{m-1}\sum_{i=1}^mn_iW_{ik}W_{il}-\frac{1}{(m-1)N}\left(\sum_{i=1}^mn_iW_{ik}\right)\left(\sum_{i=1}^mn_iW_{il}\right).
\eeqrs
Thus, we obtain
\beqrs
&&\frac{(\overline{\Sigma}-n_0\Sigma_b^0-\Sigma_\veps^0)_{k,l}}{n_0}\\
&=&\frac{1}{n_0(m-1)}\sum_{i=1}^m\{n_iW_{ik}W_{il}-(n_i\Sigma_b^0+\Sigma_\veps^0)_{k,l}\}\\
&&\qquad-\frac{1}{n_0(m-1)N}\left(\sum_{i=1}^mn_iW_{ik}\right)\left(\sum_{i=1}^mn_iW_{il}\right)\\
&&\qquad+\left\{\frac{N}{n_0(m-1)}-1\right\}(\Sigma_b^0)_{k,l}+\frac{1}{n_0(m-1)}(\Sigma^0_{\veps})_{k,l}.
\eeqrs
Then, for the first term in (\ref{eq:s.m.mdG1.0}), with $|(\Sigma_b^0)_{k,l}|\le M_b$ and $|(\Sigma^0_{\veps})_{k,l}|\le M_{\veps}$, we have
\beqr
\nonumber
&&\max_{k,l}\left|\frac{(\overline{\Sigma}-n_0\Sigma_b^0-\Sigma_\veps^0)_{k,l}}{n_0}\right|\\ \nonumber
&\le&\max_{k,l}\frac{2}{n_0m}\left|\sum_{i=1}^m\{n_iW_{ik}W_{il}-(n_i\Sigma_b^0+\Sigma_\veps^0)_{k,l}\}\right|\\
&&\qquad+\max_k\frac{2}{n_0mN}\left|\sum_{i=1}^mn_iW_{ik}\right|^2+\left\{\frac{2N}{n_0m}-1\right\}M_b+\frac{2}{n_0m}M_\veps. \label{eq:s.m.mdG1.1}
\eeqr

Recall that by assumptions $b_{ik}\in\mathcal{SG}(\sigma_{b,k}^2)$, i.e., $b_{ik}$ is sub-Gaussian with a variance factor $\sigma_{b,k}^2$ for $1\le i\le m, 1\le k\le p$. Then we have $W_{ik}\in{SG}(\sigma_{b,k}^2+n_i^{-1}\sigma_{\veps,k}^2)$. Let $\max_k \sigma_{\veps,k}^2=M_\veps$ and $\max_k \sigma_{b,k}^2=M_b$. Together with \eqref{eq:subGp1} and \eqref{eq:subGp2}, we get $\|n_iW_{ik}W_{il}-(n_i\Sigma_b^0+\Sigma_\veps^0)_{k,l}\|_{\psi_1}\le c_1(n_iM_b+M_\veps)$. Then, $n_iW_{ik}W_{il}-(n_i\Sigma_b^0+\Sigma_\veps^0)_{k,l}$ is sub-Exponential. With Bernstein's inequality, for any $k,l$, we have
\beqrs
\text{pr}\left[\left|\frac{1}{n_0m}\sum_{i=1}^m\left\{n_iW_{ik}W_{il}-(n_i\Sigma_b^0+\Sigma_\veps^0)_{k,l}\right\}\right|\ge t\right]\le 2\text{exp}\left\{-c_6\min\left(\frac{t^2n_0^2m}{K_4^2},\frac{tn_0m}{K_4}\right)\right\}, \label{eq:s.m.pbG1.0}
\eeqrs
where $c_6>0$, $K_4=\max_{i,k,l}\|n_iW_{ik}W_{il}-(n_i\Sigma_b^0+\Sigma_\veps^0)_{k,l}\|_{\psi_1}\le 2c_1n_uM_\ast$, $n_u=\max_in_i$, $M_\ast=\max(M_\veps,M_b)$.

Take $t=C_1n_u(2n_0)^{-1}(\log p/m)^{1/2}$ for a sufficiently large constant $C_1>0$. With $m\ge\log p$ and the union sum inequality, we obtain
\beqr
\nonumber
&&\text{pr}\left[\max_{k,l}\left|\frac{1}{n_0m}\sum_{i=1}^m\left\{n_iW_{ik}W_{il}-(n_i\Sigma_b^0+\Sigma_\veps^0)_{k,l}\right\}\right|\ge t\right]\\ \nonumber
&\le& 2p^2\text{exp}\left\{-c_6\min\left(\frac{t^2n_0^2m}{4c_1^2n_u^2M_\ast^2},\frac{tn_0m}{2c_1n_uM_\ast}\right)\right\}\\ \nonumber
&=&2\text{exp}\left[2\log p-\min\left\{\frac{c_6C_1^2}{16c_1^2M_\ast^2}\log p,\frac{c_6C_1}{4c_1M_\ast}(m\log p)^{1/2}\right\}\right]\\ 
&\le&2\text{exp}\left[\left\{2-\min\left(\frac{c_6C_1^2}{16c_1^2M_\ast^2},\frac{c_6C_1}{4c_1M_\ast}\right)\right\}\log p\right]. \label{eq:s.m.pbG1.4}
\eeqr

Then we will bound the second term in (\ref{eq:s.m.mdG1.1}). By the property of sub-Gaussian assumption, $n_iW_{ik}=n_ib_{ik}+\sum_{j=1}^{n_i}\veps_{ijk}\in\mathcal{SG}(n_i^2M_b+n_iM_{\veps})$. Then, according to the general Hoeffding’s inequality (Theorem 2.6.2 by \citet{Vershyin:2018}), we have
\beqrs
\nonumber
\text{pr}\left(\frac{1}{n_0mN}\left|\sum_{i=1}^mn_iW_{ik}\right|^2\ge t\right)&\le&\text{pr}\left\{\left|\sum_{i=1}^mn_iW_{ik}\right|\ge (tn_0mN)^{1/2}\right\}\\ \nonumber
&\le&2\text{exp}\left\{-\frac{c_7tn_0mN}{\sum_{i=1}^mc_0^2(n_i^2M_b+n_iM_{\veps})}\right\}\\ \nonumber
&\le&2\text{exp}\left\{-\frac{c_7tn_0mN}{\sum_{i=1}^mc_0^2(n_in_uM_b+n_iM_{\veps})}\right\}\\
&\le&2\text{exp}\left(-\frac{c_7tn_0m}{c_1n_uM_\ast}\right), \label{eq:s.m.pbG1.5}
\eeqrs
where $c_7>0$.

Then, take $t=C_1n_u(2n_0)^{-1}(\log p/m)^{1/2}$, with $m\ge\log p$, by the union sum inequality, we have
\beqr
\nonumber
\text{pr}\left(\max_{k}\frac{1}{n_0mN}\left|\sum_{i=1}^mn_iW_{ik}\right|^2\ge t\right)&\le& 2p\text{exp}\left(-\frac{c_7tn_0m}{c_1n_uM_\ast}\right)\\ \nonumber
&\le& 2p\text{exp}\left\{-\frac{c_7C_1}{2c_1M_\ast}(m\log p)^{1/2}\right\}\\
&\le&2\text{exp}\left\{\left(1-\frac{c_7C_1}{2c_1M_\ast}\right)\log p\right\} \label{eq:s.m.pbG1.7}
\eeqr
for a sufficiently large constant $C_1>0$.

To bound the second term in \eqref{eq:s.m.mdG1.0}, by Lemma \ref{lem:max.r}, for a sufficiently large constant $C_2>0$, we have
\beqr
\text{pr}\left\{\max_{k,l}\left|\frac{(\widehat{\Sigma}_\veps-\Sigma_\veps^0)_{k,l}}{n_0}\right|>2C_2\frac{(N\log p )^{1/2}}{n_0(N-m)}\right\}\le4p^{-C_3^{\prime}}, \label{eq:s.m.pbG1.8}
\eeqr
where $C_3^{\prime}>0$ only depends on $C_2$ and $M_\veps$. 

Then, with 
$$\widetilde{\lambda}_b=C_1\frac{\max n_i}{n_0}\left(\frac{\log p}{m}\right)^{1/2}+C_2\frac{(N\log p)^{1/2}}{n_0(N-m)}+\frac{(2N-n_0m)M_b}{2n_0m}+\frac{M_\veps}{n_0m}$$ 
for sufficiently large $C_1, C_2>0$, combining \eqref{eq:s.m.mdG1.0}-\eqref{eq:s.m.pbG1.8}, we obtain
\beqrs
\nonumber
&&\text{pr}\left\{\max_{k,l}\left|(\widetilde{\Sigma}_b-\Sigma_b^0)_{k,l}\right|>2\widetilde{\lambda}_b\right\}\\ \nonumber
&\le&\text{pr}\left[\max_{k,l}\left|\frac{2}{n_0m}\sum_{i=1}^m\left\{n_iW_{ik}W_{il}-(n_i\Sigma_b^0+\Sigma_\veps^0)_{k,l}\right\}\right|\ge C_1\frac{\max_in_i}{n_0}\left(\frac{\log p}{m}\right)^{1/2}\right]\\ \nonumber
&&\qquad+\text{pr}\left\{\max_{k}\frac{2}{n_0mN}\left|\sum_{i=1}^mn_iW_{ik}\right|^2\ge C_1\frac{\max_in_i}{n_0}\left(\frac{\log p}{m}\right)^{1/2}\right\}\\ \nonumber
&&\qquad+\text{pr}\left\{\max_{k,l}\left|\frac{(\widehat{\Sigma}_\veps-\Sigma_\veps^0)_{k,l}}{n_0}\right|>2C_2\frac{(N\log p )^{1/2}}{n_0(N-m)}\right\}\\ \nonumber
&\le&4p^{-C_3^{\dprime}}+4p^{-C_3^\prime}\\
&\le&8p^{-C_3}, \label{eq:s.m.pbG1.9}
\eeqrs
where $C_3^{\dprime}=\min\{c_6C_1^2(16c_1^2M_\ast^2)^{-1},c_6C_1(4c_1M_\ast)^{-1},c_7C_1(2c_1M_\ast)^{-1}+1\}-2$ and $C_3 = \min\{C_3^\prime,C_3^{\dprime}\}$.
\end{proof}

\begin{lemma}\label{lem:max.ae}
Consider the true within-subject covariance $\Sigma^0_{\veps}$ with  $\max_k (\Sigma^0_\veps)_{k,k}\le M_{\veps}$ and the true between-subject covariance $\Sigma^0_b$ with  $\max_k (\Sigma^0_b)_{k,k}\le M_b$. Let
\beqrs
\lambda_1 = C_1 \left(\frac{\log p}{m}\right)^{1/2} + M_b + \frac{(2-n^\ast)M_\veps}{2n^\ast } 
\eeqrs
for sufficiently large $C_1>0$, where $n^\ast=m/\sum_{i=1}^mn_i^{-1}$. If $\log p\le m$, then $\overline{\Sigma}$ satisfies 
\beqrs
\text{pr}\left\{\max_{k,l}\left|(\overline{\Sigma}-\Sigma_\veps^0)_{k,l}\right|>2\lambda_1\right\}\le4p^{-C_2}
\eeqrs
where $C_2> 0$ only depends on $C_1$ and $\max(M_\veps,M_b)$.
\end{lemma}

\begin{proof}
Now, we will  consider the convergence rate of $\max_{k,l}|(\overline{\Sigma}-\Sigma^0_\veps)_{k,l}|$. Note that $(\overline{\Sigma}-\Sigma^0_\veps)_{k,l}=(\overline{\Sigma}-\Sigma^0_b)_{k,l}+(\Sigma^0_b)_{k,l}-(\Sigma^0_\veps)_{k,l}$. Then, by \eqref{eq:a.dc.1}, with $|(\Sigma_b^0)_{k,l}|\le M_b$ and $|(\Sigma^0_{\veps})_{k,l}|\le M_{\veps}$, we have
\beqr
\nonumber
\max_{k,l}\left|(\overline{\Sigma}-\Sigma^0_\veps)_{k,l}\right|&\le&2\max_{k,l}\left|\frac{1}{m}\sum_{i=1}^m\left\{W_{ik}W_{il}-\left(\Sigma_b^0+n_i^{-1}\Sigma^0_{\veps}\right)_{k,l}\right\}\right|\\ \nonumber
&&+2\max_{k,l}\left|\left(\frac{1}{m}\sum_{i=1}^mW_{ik}\right)\left(\frac{1}{m}\sum_{i=1}^mW_{il}\right)\right|\\
&&+2M_b+\frac{(2-n^\ast)M_\veps}{n^\ast}. \label{eq:a.aec.2}
\eeqr

Following the steps in Lemma \ref{lem:max.b2}, with
\beqrs
\lambda_1=C_1\left(\frac{\log p}{m}\right)^{1/2}+M_b+\frac{(2-n^\ast)M_\veps}{2n^\ast}
\eeqrs
for a sufficiently large constant $C_1>0$, we have
\beqrs
\text{pr}\left\{\max_{k,l}\left|(\overline{\Sigma}-\Sigma_\veps^0)_{k,l}\right|>2\lambda_1\right\}\le4p^{-C_2},
\eeqrs
where $C_2>0$ only depends on $C_1$ and $\max(M_\veps,M_b)$.
\end{proof}

\subsection{Proof of Theorem 1}
\begin{proof}
Define $\Delta_{\veps}=\Sigma_\veps-\Sigma^0_{\veps}$ and $F_\veps(\Delta_{\veps})=\|\Delta_{\veps}+\Sigma^0_{\veps}-\widehat{\Sigma}_\veps\|_F^2/2+\lambda_\veps|\Delta_{\veps}+\Sigma^0_{\veps}|_1$, then the objective function \eqref{obj:admm0} is equivalent to 
\beqrs
\min_{\Delta_{\veps}:\Delta_{\veps}=\Delta_{\veps}^\T,\Delta_{\veps}+\Sigma^0_{\veps}\succeq\delta I}F_\veps(\Delta_{\veps}).
\eeqrs
Consider the set
\beqr
\{\Delta_{\veps}:\Delta_{\veps}=\Delta_{\veps}^\T,\Delta_{\veps}+\Sigma^0_{\veps}\succeq\delta I,\|\Delta_{\veps}\|_F=5\lambda_\veps(ps_\veps)^{1/2}\}. \label{eq:set_r}
\eeqr 
According to \cite{Xue.etal:2012}, under the probability event $\{|(\widehat{\Sigma}_\veps-\Sigma_\veps^0)_{,k,l}|\le\lambda_\veps,\forall(i,j)\}$, we have
\beqrs
F_\veps(\Delta_{\veps})-F_\veps(\bm{0})&\ge&\frac{1}{2}\|\Delta_{\veps}\|_F^2-2\lambda_\veps\left[\sum_{k,l=1}^p1\{(\Sigma_\veps^0)_{k,l}\ne0\}\right]^{1/2}\|\Delta_{\veps}\|_F\\
&\ge&\frac{1}{2}\|\Delta_{\veps}\|_F^2-2\lambda_\veps(ps_\veps)^{1/2}\|\Delta_{\veps}\|_F\\
&=&\frac{5}{2}\lambda_\veps^2ps_\veps\\
&>&0.
\eeqrs
Note that $F_\veps(\Delta_{\veps})$ is a convex function and $F_\veps(\widehat{\Delta}_\veps)\le F_\veps(\bm{0})=0$. Then, the minimizer $\widehat{\Delta}_\veps$ must be inside the sphere (\ref{eq:set_r}). Hence, we have
\beqrs
\nonumber
&&\text{pr}\left\{\left\|\widehat{\Sigma}_\veps^{+}-\Sigma^0_{\veps}\right\|_F\le5\lambda_\veps(ps_\veps)^{1/2}\right\}\\ \nonumber
&\ge&1-\text{pr}\left\{\max_{k,l}\left|(\widehat{\Sigma}_\veps-\Sigma_\veps^0)_{k,l}\right|>\lambda_\veps\right\}\\
&\ge&1-4p^{-C_2}. \label{eq:s.m.pbR0.0}
\eeqrs
\end{proof}

\setcounter{theorem}{4}
\begin{theorem}\label{thm:max.ae.c}
Consider the true between-subject covariance matrix $\Sigma^0_{b} \in \mathcal{U}(M_{b}, s_{b})$ and the true within-subject covariance matrix $\Sigma^0_{\veps} \in \mathcal{U}(M_{\veps}, s_{\veps})$. Let
\beqrs
\lambda_1 = C_1 \left(\frac{\log p}{m}\right)^{1/2} 
+ M_b +   \frac{(2-n^\ast)M_\veps}{2n^\ast}
\eeqrs
be the value of the tuning parameter $\lambda$ in \eqref{obj:admm0} for sufficiently large $C_1 > 0$, and the same $n^\ast$ defined in Theorem 2.
If $\log p\le m$, then the naive estimator $\overline{\Sigma}^+$ satisfies 
$$
\left\|\overline{\Sigma}^{+}-\Sigma^0_\veps\right\|_F\le10\lambda_1 (ps_\veps)^{1/2} %\label{eq:thm.g2matrix.bound.c}
$$
with probability at least $1- 4p^{-C_2}$, where $C_2> 0$ only depends on $C_1$ and $\max(M_\veps,M_b)$. 
\end{theorem}

The proof of Theorem 2, 3, 4 and 5 follow straightforwardly from Theorem 1.

\subsection{Corollaries}
\label{sub:cor}
In some scenarios, estimation of between-subject and within-subject correlation matrices (instead of the covariance matrices) is of interest. And the sparse and positive-definite estimate of $R_\veps$, denoted as $\widehat{R}_\veps^+$, and of $\Sigma_b$, denoted as $\widehat{R}_b^+$ are defined as solution of \eqref{obj:admm0} with $B=\widehat{R}_\veps=D_\veps^{-1/2}\widehat{\Sigma}_\veps D_\veps^{-1/2}$ and $B=\widehat{R}_b=D_b^{-1/2}\widehat{\Sigma}_b D_b^{-1/2}$, where $D_\veps=\text{diag}\{(\widehat{\Sigma}_\veps)_{1,1},\ldots,(\widehat{\Sigma}_\veps)_{p,p}\}$ and $D_b=\text{diag}\{(\widehat{\Sigma}_b)_{1,1},\ldots,(\widehat{\Sigma}_b)_{p,p}\}$.

\begin{corollary}\label{cor:max.r.0.1}
Under conditions of Theorem 1, if $\min_k(\Sigma_\veps^0)_{k,k}$ is bounded from below, then 
\beqrs
\left\|\widehat{R}_\veps^+-R^0_{\veps}\right\|_F=O_P\left\{\frac{(ps_\veps N\log p)^{1/2}}{N-m}\right\}, \label{eq:thm.rmatrix.bound.c1.1}
\eeqrs
uniformly on $\Sigma^0_{\veps}\in\mathcal{U}(M_\veps,s_\veps)$, as $N,m\to\infty$.
\end{corollary}

\begin{proof} 
By Lemma \ref{lem:max.r}, we have
\beqr
\text{pr}\left\{\max_{k,l}\left|(\widehat{\Sigma}_\veps-\Sigma_\veps^0)_{k,l}\right|>C_1\frac{(N\log p)^{1/2}}{N-m}\right\}=o(1). \label{eq:s.m.pbRR0}
\eeqr

By Lemma 2 in \citet{Cui.etal:2016}, with \eqref{eq:s.m.pbRR0} and the fact $(\widehat{R}_\veps)_{k,l}=(\widehat{\Sigma}_\veps)_{k,l}/\{(\widehat{\Sigma}_\veps)_{k,k}(\widehat{\Sigma}_\veps)_{l,l}\}^{1/2}$, for a sufficiently large constant $C_1^\prime>0$, we have
\beqrs
\text{pr}\left\{\max_{k,l}\left|(\widehat{R}_\veps-R_\veps^0)_{k,l}\right|>C_1^\prime\frac{(N\log p)^{1/2}}{N-m}\right\}=o(1). \label{eq:s.m.pbRR1}
\eeqrs

Following the steps in proof of Theorem 1, it is easily shown that 
\beqrs
\left\|\widehat{R}_\veps^+-R^0_{\veps}\right\|_F=O_P\left\{\frac{(ps_\veps N\log p)^{1/2}}{N-m}\right\}.
\eeqrs
\end{proof}

\begin{corollary}\label{cor:max.g.0.1}
Under conditions of Theorem 2, if $\min_k(\Sigma_\veps^0)_{k,k}$ and $\min_k(\Sigma_b^0)_{k,k}$ are bounded from below, then 
\beqrs
\left\|\widehat{R}_{b}^+-R^0_b\right\|_F=O_P\left[(ps_b)^{1/2}\left\{C_1^\prime\left(\frac{\log p}{m}\right)^{1/2}+C_2^\prime\frac{(N\log p)^{1/2}}{(N-m)n_0}\right\}\right], \label{eq:thm.gmatrix.bound.c1.0}
\eeqrs
uniformly on $\Sigma^0_{\veps}\in\mathcal{U}(M_\veps,s_\veps)$ and $\Sigma^0_b\in\mathcal{U}(M_b,s_b)$, for some large $C_1^\prime, C_2^\prime>0$, as $m,n\to\infty$.
\end{corollary}

\section{More Results on Numerical Study}
\label{sec:more_sim}
\subsection{Unconstrained estimators versus constrained estimators}
\label{sec:2compares}
We generate 100 independent data sets for both balanced Model 1 and Model 2 with $n_i=2$ and $m=100$. We compare the performance of the unconstrained estimators, $\mathcal{S}_{\lambda}(\widehat{\Sigma}_\veps)$ and $\mathcal{S}_{\lambda}(\widehat{\Sigma}_{b})$, and the constrained estimators, $\widehat{\Sigma}_\veps^+$ and $\widehat{\Sigma}_{b}^+$, in terms of  estimation errors and the percentage of positive definite estimators, where $\mathcal{S}_{\lambda}()$ is the soft-thresholding operator defined in Section \ref{sec:method}. The simulation results are summarized in Table \ref{tab:error}. In general, the constrained estimators have slightly better performance in terms of estimation errors. In addition, we demonstrate that the positive definite constraint is crucial by observing that in most of the cases the unconstrained estimator are not guaranteed to be positive definite, making them less qualified for interpretation or downstream statistical tasks. 

\begin{table}[H]
\caption{Comparison of the unconstrained and constrained estimators under the balanced setting. Each metric is averaged over 100 replicates with the standard error shown in the parentheses. Comparisons are in terms of the estimation errors ($F$-error and $L_2$-error) and the percentage of positive definite estimators.}
\begin{tabular}{llcccc}
    & &\multicolumn{2}{c}{Model 1} & \multicolumn{2}{c}{Model 2}\\
    & $p$ & 100 & 200 & 100 & 200\\
    \multicolumn{6}{c}{Within-Subject}\\
    \multirow{2}{*}{$F$-error}      & $\mathcal{S}_{\lambda}(\widehat{\Sigma}_\veps)$ & 7.1804 (0.0562) & 11.4040 (0.0490) & 5.3956 (0.0202) & 8.3116 (0.0159)\\
                                    & $\widehat{\Sigma}_\veps^+$                      & 7.0548 (0.0552) & 11.1804 (0.0490) & 5.3956 (0.0202) & 8.3116 (0.0159)\\
    \multirow{2}{*}{$L_2$-error}    & $\mathcal{S}_{\lambda}(\widehat{\Sigma}_\veps)$ & 3.6179 (0.0451) & 4.2217 (0.0282)  & 2.7131 (0.0115) & 2.1257 (0.0083)\\
                                    & $\widehat{\Sigma}_\veps^+$                      & 3.5553 (0.0438) & 4.1564 (0.0286)  & 2.7131 (0.0115) & 2.1257 (0.0083)\\
    \multirow{2}{*}{PD\%}           & $\mathcal{S}_{\lambda}(\widehat{\Sigma}_\veps)$ & 18\%            & 3\%              & 100\%           & 100\%          \\
                                    & $\widehat{\Sigma}_\veps^+$                      & 100\%           & 100\%            & 100\%           & 100\%          \\
    \multicolumn{6}{c}{Between-Subject}\\
    \multirow{2}{*}{$F$-error}      & $\mathcal{S}_{\lambda}(\widehat{\Sigma}_b)$    & 10.8195 (0.0611) & 17.0538 (0.0416) & 7.6064 (0.0212)  & 11.6116 (0.0187)\\
                                    & $\widehat{\Sigma}_{b}^+$                       & 10.1304 (0.0635) & 16.1446 (0.0436) & 7.5382 (0.0222)  & 11.6005 (0.0139)\\
    \multirow{2}{*}{$L_2$-error}    & $\mathcal{S}_{\lambda}(\widehat{\Sigma}_b)$    & 4.5419 (0.0447)  & 5.3739 (0.0258)  & 2.3508 (0.0104)  & 2.5681 (0.0051) \\
                                    & $\widehat{\Sigma}_{b}^+$                       & 4.2857 (0.0467)  & 5.0994 (0.0257)  & 2.3143 (0.0104)  & 2.5358 (0.0046) \\
    \multirow{2}{*}{PD\%}           & $\mathcal{S}_{\lambda}(\widehat{\Sigma}_b)$    & 0\%              & 0\%              & 7\%              & 12\%            \\
                                    & $\widehat{\Sigma}_{b}^+$                       & 100\%            & 100\%            & 100\%            & 100\%           \\
\end{tabular}
\label{tab:error}
\begin{tablenotes}
      \small 
      \item $\mathcal{S}_{\lambda}(\widehat{\Sigma}_\veps)$ and $\widehat{\Sigma}_\veps^+$: unconstrained and constrained estimators for within-subject covariance; $\mathcal{S}_{\lambda}(\widehat{\Sigma}_b)$ and $\widehat{\Sigma}_{b}^+$: unconstrained and constrained estimators for between-subject covariance; PD\%, percentage of positive definite estimators.
\end{tablenotes}
\end{table}
\end{appendices}

%\baselineskip=16pt
%\bibliography{reference}
%\bibliographystyle{dcu}
%\citationstyle{dcu}
%\clearpage\pagebreak\newpage
\addcontentsline{toc}{section}{Bibliography}
\renewcommand{\bibsection}{

\end{document}